\documentclass[lettersize,journal]{IEEEtran}
\usepackage{amsmath,amssymb,amsfonts}
\usepackage{mathtools}
\usepackage{upgreek}
\usepackage{cases}
\usepackage{xintexpr}
\usepackage{pifont}
\usepackage{ntheorem}
\usepackage{booktabs}
\usepackage{multirow}
\usepackage{graphicx}
\usepackage{subcaption}
\usepackage{overpic}
\usepackage{textcomp}
\usepackage{xcolor}
\usepackage{xifthen}
\usepackage{multicol}
\usepackage{cuted}
\usepackage{cite}
\usepackage{enumitem}
\usepackage[switch]{lineno}
\usepackage[numbers,sort&compress]{natbib}
\usepackage[ruled,linesnumbered,noend]{algorithm2e}
\usepackage{hyperref,soul}
\DeclareMathAlphabet{\mathpzc}{OT1}{pzc}{m}{it}

\def\BibTeX{{\rm B\kern-.05em{\sc i\kern-.025em b}\kern-.08em
    N\kern-.1667em\lower.7ex\hbox{E}\kern-.125emX}}

\newcommand{\code}[1][k]{\boldsymbol{b}^{(#1)}}

\newcommand{\coeff}[1][]{
\ifthenelse{\isempty{#1}}
{\boldsymbol{\gamma}}
{\boldsymbol{\gamma}^{(#1)}}}

\newcommand{\de}[1][]{
\ifthenelse{\isempty{#1}}
{\Delta}{\Delta^{(#1)}}}
\newcommand{\pde}[1][]{
\ifthenelse{\isempty{#1}}
{\widetilde{\Delta}}{\widetilde{\Delta}^{(#1)}}}
\newcommand{\hde}[1][]{
\ifthenelse{\isempty{#1}}
{\widehat{\Delta}}{\widehat{\Delta}^{(#1)}}}

\newcommand{\mse}[1][]{
\ifthenelse{\isempty{#1}}
{J}
{J^{(#1)}}}

\newcommand{\quanterr}[1][]{
\ifthenelse{\isempty{#1}}
{\boldsymbol{\zeta}}
{\boldsymbol{\zeta}^{(#1)}}}

\newcommand{\grad}[1][k]{\boldsymbol{g}_m^{(#1)}}
\newcommand{\stograd}[1][k]{\tilde{\boldsymbol{g}}_m^{(#1)}}

\newcommand{\z}{\boldsymbol{z}}

\newcommand{\vecA}{\boldsymbol{a}}
\newcommand{\vecB}{\boldsymbol{b}}

\newcommand{\firmoment}[1][]{
\ifthenelse{\isempty{#1}}
{\boldsymbol{u}}
{\boldsymbol{u}^{(#1)}}}
\newcommand{\firmom}[1][]{
\ifthenelse{\isempty{#1}}
{u}
{u^{(#1)}}}
\newcommand{\secmoment}[1][]{
\ifthenelse{\isempty{#1}}
{\boldsymbol{v}}
{\boldsymbol{v}^{(#1)}}}
\newcommand{\secmom}[1][]{
\ifthenelse{\isempty{#1}}
{v}
{v^{(#1)}}}

\newcommand{\denom}[1][]{
\ifthenelse{\isempty{#1}}
{\boldsymbol{\nu}}
{\boldsymbol{\nu}^{(#1)}}}
\newcommand{\den}[1][]{
\ifthenelse{\isempty{#1}}
{\nu}
{\nu^{(#1)}}}

\newcommand{\hw}[1][]{
\ifthenelse{\isempty{#1}}
{\mathbf{w}}
{\widehat{\mathbf{w}}^{(#1)}}}

\newcommand{\eps}[1][]{
\ifthenelse{\isempty{#1}}
{\boldsymbol{\varepsilon}}
{\boldsymbol{\varepsilon}^{(#1)}}}
\newcommand{\epsidx}[1][]{
\ifthenelse{\isempty{#1}}
{\varepsilon^{(t)}}
{\varepsilon^{(t)}_{#1}}}

\newcommand{\qw}[1][]{
\ifthenelse{\isempty{#1}}
{\mathbf{w}}
{\widetilde{\mathbf{w}}^{(#1)}}}

\newcommand{\pw}[1][]{
\ifthenelse{\isempty{#1}}
{\mathbf{w}}
{\widetilde{\mathbf{w}}^{(#1)}}}

\newcommand{\res}[1][]{  
\ifthenelse{\isempty{#1}}{\mathbf{e}}{\mathbf{e}^{(#1)}}}
\newcommand{\hres}[1][k]{\hat{\mathbf{e}}^{(#1)}}

\newcommand{\w}[1][]{
\ifthenelse{\isempty{#1}}
{\mathbf{w}}
{\mathbf{w}^{(#1)}}}

\newcommand{\enc}{\operatorname{Enc}}
\newcommand{\dec}{\operatorname{Dec}}

\newcommand*{\diff}{\mathop{}\!\mathrm{d}}

\newcommand{\eenc}{\operatorname{entropy\_enc}}
\newcommand{\edec}{\operatorname{entropy\_dec}}

\newcommand{\expect}[1][]{
\ifthenelse{\isempty{#1}}
{\mathbb{E}}
{\mathbb{E}\left[#1\right]}}

\newcommand{\innerprod}[2]
{\left\langle #1,\,#2 \right\rangle}
\newcommand{\innprod}[2]
{\Big\langle #1,\,#2 \Big\rangle}

\newcommand{\normsq}[1]
{\big\| #1\big\|_2^2}
\newcommand{\normSQ}[1]
{\left\| #1\right\|_2^2}

\newcommand{\order}[1]{\mathcal{O}\left( #1 \right)}

\newcommand{\pred}{\operatorname{pred}}

\newcommand{\prob}{\mathbb{P}}

\newcommand{\pdist}[1][m]{\mathcal{P}_{m}}

\newcommand{\quant}{\operatorname{quant}}
\newcommand{\dequant}{\operatorname{dequant}}
\newcommand{\stoquant}{Q_{\text{s}}}
\newcommand{\uniformquant}{Q_{\text{u}}}

\newcommand{\univar}{\varphi^{\textrm{u}}}
\newcommand{\stovar}{\varphi^{\textrm{s}}}
\newcommand{\sym}{\boldsymbol{h}}

\newcommand{\sign}{\operatorname{sign}}

\newcommand{\cirone}
{\text{\ding{172}}}
\newcommand{\cirtwo}
{\text{\ding{173}}}
\newcommand{\cirthree}
{\text{\ding{174}}}
\newcommand{\cirfour}
{\text{\ding{175}}}

\newcommand{\intab}[2][0.75]{\scalebox{#1}{\textrm{#2}}}

\newenvironment{proof}[1][]{
\ifthenelse{\isempty{#1}}
{\par\vspace*{-1mm}\noindent\textit{Proof.} }
{\par\vspace*{-2mm}\noindent\textit{Proof of #1.} }}
{\hfill$\square$ \vspace*{2mm}}

\newtheorem{Assumption}{Assumption}
\newtheorem{Theorem}{Theorem}
\newtheorem{Lemma}{Lemma}
\newtheorem{Remark}{Remark}
\theoremheaderfont{\rmfamily}

\newcommand{\highlight}[1]{\vspace{1mm}\noindent{}\textbf{#1}\hspace{3mm}}
\newcommand{\papertheorem}[1]{Theorem~\ref{#1} }
\newcommand{\paperfig}[1]{Fig.~\ref{#1}}
\newcommand{\papertab}[1]{TABLE~\ref{#1}}

\newcommand{\paperApp}[1]{Appendix~\ref{#1}}

\begin{document}
\title{Communication-Efficient Federated Learning \\ via Predictive Coding}
\author{Kai Yue, \textit{Graduate Student Member, IEEE}, Richeng Jin, \textit{Member, IEEE}, \\ Chau-Wai Wong, \textit{Member, IEEE}, and Huaiyu Dai, \textit{Fellow, IEEE}
\thanks{Manuscript received July 31, 2021; revised December 15, 2021; accepted Jan 4, 2022.
The guest editor coordinating the
review of this manuscript and approving it for publication was Prof. Zhu Han. }}

\markboth{IEEE JOURNAL OF SELECTED TOPICS IN SIGNAL PROCESSING}%
{Communication-Efficient Federated Learning via Predictive Coding}

\maketitle
\makeatletter{\renewcommand*{\@makefnmark}{}
\footnotetext{The authors are with the Department of Electrical and Computer Engineering, NC State University, Raleigh, NC 27695 USA (e-mail: \{kyue, rjin2, chauwai.wong, hdai\}@ncsu.edu). }
\makeatother}

\begin{abstract}
Federated learning can enable remote workers to collaboratively train a shared machine learning model while allowing training data to be kept locally.
In the use case of wireless mobile devices, the communication overhead is a critical bottleneck due to limited power and bandwidth. 
Prior work has utilized various data compression tools such as quantization and sparsification to reduce the overhead. 
In this paper, we propose a predictive coding based compression scheme for federated learning. 
The scheme has shared prediction functions among all devices and allows each worker to transmit a compressed residual vector derived from the reference.
In each communication round, we select the predictor and quantizer based on the rate--distortion cost, and further reduce the redundancy with entropy coding.
Extensive simulations reveal that the communication cost can be reduced up to 99\% with even better learning performance when compared with other baseline methods.
\end{abstract}


\section{Introduction}
Machine learning has achieved unprecedented success in recent years with the availability of big data and increased computational power. 
One notable example is deep learning, which uses neural networks of a large number of hidden layers and parameters efficiently trained on enormous labeled data. 
To leverage distributed training sets available at edge devices while simultaneously protecting the privacy of their data, a new paradigm named federated learning  has been developed~\cite{mcmahan2017communication, kairouz2019advances}. 
In federated learning, multiple workers communicate with a server and solve a machine learning task under its coordination. 
This approach allows training a joint model collaboratively without the need to share private data among edge devices.

One major challenge in federated learning is the expensive communication cost. 
Commonly used optimization methods such as the stochastic gradient descent (SGD) require many rounds of communication within the distributed network. 
Considering the large number of parameters of modern neural networks and the limited resources of edge devices, the recurrent burden of communication can be a primary impediment if the joint training does not converge within a short time frame. 
For instance, the well-known bidirectional encoder representations from transformers (BERT) model contains $110$ million parameters~\cite{devlin2019bert}. 
If such a model is trained via federated learning, the accumulated communication costs will become formidable before the training terminates.

Various methods have been proposed to lower communication overhead in federated learning.
One direction is to reduce the data volume of model update in each communication round by quantizing~\cite{alistarh2017qsgd} or sparsifying the gradient~\cite{sattler2019robust}. 
Another direction is to reduce the number of communication rounds by using periodic communication strategies~\cite{mcmahan2017communication, stich2019local}. 
A popular implementation is federated averaging (FedAvg)~\cite{mcmahan2017communication}, where the broadcasted model from the server will be locally updated on workers for a couple of iterations between successive communication rounds. 
However, directly quantizing the model weights or gradients may not provide the best trade-off between communication efficiency and model utility. 
In this paper, we propose a predictive coding based compression scheme to diversify the designs of communication-efficient federated learning algorithms.
Specifically, we utilize shared prediction functions among devices to decorrelate the successive updates via predictive coding. 
The resulting residual vectors are subsequently quantized and entropy coded to save the uplink cost.
Our contributions are summarized as follows.
\begin{enumerate}[itemindent=2ex]
    \item We design a compression scheme for FedAvg-type algorithms~\cite{mcmahan2017communication}. 
          Compared to existing methods directly compressing gradients, our method reduces the communication cost with even better learning performance.
    \item To the best of our knowledge, we are among the first to exploit the predictive coding tool to improve the communication efficiency of federated learning. 
    \item The proposed scheme can be viewed as a general compression method for federated learning. 
    Our designed predictive coding, quantization, and entropy coding components can be incorporated into other existing FedAvg-type algorithms jointly or separately.      
\end{enumerate}

The remainder of the paper is organized as follows. 
In Section~\ref{section:related_works}, we review the relevant work. 
In Sections~\ref{section:system_model}--\ref{section:algorithm_analysis}, we formulate the problem, present the proposed scheme, and conduct the analysis of the algorithm, respectively.
We discuss the experimental results in Section~\ref{section:experiment} and conclude the paper in Section~\ref{conclusion_section}

\section{Related Work}\label{section:related_works}

\subsection{Predictive Coding as a Compression Tool}
Predictive coding has been widely adopted in signal compression for decades. 
In the predictive coding framework, predictors are designed to estimate a target signal from past or present observations.
A residual signal can be obtained based on the output of the predictor, which tends to have a distribution with lower entropy compared with the initial distribution.  
Eliminating statistical redundancy via predictive coding is one of the key steps in data compression~\cite{spanias2006audio}. 
Researchers have empirically verified the effectiveness of the predictive coding tool in the application of video coding~\cite{sze2014high}. 
Prior work also analytically shows that quantizing the raw input signal rather than a decorrelated version will lower the rate--distortion performance~\cite{wong2016impact}. 

Predictive tools can be traced back to the differential pulse code modulation (DPCM) in the 1950s~\cite{sayood2017introduction}.  
For analog-to-digital (A/D) conversion, DPCM can reduce the bandwidth by transmitting the difference between two successive sample values~\cite{lathi1998modern}. 
Adaptive DPCM introduces dynamic coefficients that are adjusted based on the prediction error. 
Different from waveform coding such as DPCM, linear predictive coding (LPC) is developed for speech signal synthesis. 
LPC analyzer estimates the voice signal formants, which can be removed to calculate the residual signal~\cite{lathi1998modern}. 
For image compression, JPEG contains different predictive schemes to perform the prediction based on adjacent pixel values~\cite{sayood2017introduction}. 
All possible different predictions are tried and the one achieving the lowest bitrate is used.
Sophisticated video coding frameworks such as HEVC/H.265 exploit different predictive tools to reduce the spatial and temporal redundancy and enable the compression of the video data~\cite{sze2014high}. 

To remove the signal redundancy for compression, predictors are designed based on different domain knowledge to adapt to various use cases.
In our work, we incorporate the idea of predictive coding into the design of a federated learning scheme to compress the uplink data needed for the joint training. 
We construct several predictors to estimate the model updates and then compress the residual vector between the original signal and the estimation.

\subsection{Efficient Communication in Federated Learning}
In federated learning, the main focus of reducing the communication cost is on the uplink transmission, as the downlink bandwidth is much larger and the server is assumed to have enough transmission power~\cite{tran2019federated}.  
Prior studies have proposed different strategies to reduce the size of the transmitted message in each communication round. 
Wiedemann et al. \cite{wiedemann2020deepcabac} designed a coding scheme to compress the neural network weights. 
Other works applied quantization or sparsification tools to reduce the size of the gradients~\cite{sattler2019robust,reisizadeh2020fedpaq}. 
To reduce the negative effect introduced by gradient compression, recent works have also developed error feedback mechanisms~\cite{liu2020double}.

From the coding theory perspective~\cite{sayood2017introduction}, existing works that directly compressed the weights or gradients to be transmitted may not reduce coding rate in the most efficient way.
In this work, we demonstrate that encoding the prediction residue of the model weight provides a better trade-off between communication efficiency and model utility. 
We show that compressing the model update, as proposed in~\cite{haddadpour2020fedcom, reisizadeh2020fedpaq}, is analogous to the DPCM encoding of the weight in a general coding scheme.

\subsection{Efficient Communication with Entropy Coding}
Entropy coding is a lossless compression tool and it typically follows the lossy quantization step~\cite{sayood2017introduction}.
The combination allows a smooth trade-off between the bitrate and data fidelity.  
Previous research has utilized different entropy coding techniques to ease the communication burden. 
Quantized stochastic gradient descent (QSGD)~\cite{alistarh2017qsgd} combines the quantization and entropy coding to compress gradient vectors. 
In their work, a sparse gradient is generated with a stochastic quantizer, and the positions of nonzero entries are compressed with Elias integer coding~\cite{elias1975universal}. 
In~\cite{sattler2019sparse, sattler2019robust}, the authors first sparsify the gradient by selecting entries with large magnitudes. The distances between nonzero elements are then encoded with Golomb coding~\cite{sayood2017introduction}.  
Likewise, Lin et al.~\cite{lin2018deep} set a threshold to filter out gradients with small magnitudes and adopt run-length coding~\cite{sayood2017introduction} to encode the sparsified vector.
Compared to prior studies~\cite{alistarh2017qsgd, lin2018deep, sattler2019robust, sattler2019sparse} that encode nonzero entries in model updates, our work adopts an entropy coder to reduce the average codeword length for the quantized residues.

\section{System Model}\label{section:system_model}
\subsection{Federated Learning Model}
Consider a federated learning architecture where a server trains a global model by exchanging information with $M$ workers, each of which hosts a private local dataset. 
The local dataset of the $m$th worker is denoted as $\mathcal{D}_{m} = \{\z_{m,j}\}_{j=1}^{n_m}$, 
where $\z_{m,j}$ is the $j$th data point, containing a pair of input and label, drawn from a distribution $\pdist$.
The local objective can be formulated as the empirical risk function with finite data points: 
\begin{equation}\label{def:local_objective_function}
    f_{m}(\w) \triangleq f_{m}(\w; \mathcal{D}_{m}) = \frac{1}{n_m} \sum_{j=1}^{n_m} \ell(\w; \z_{m,j}),
\end{equation}
where $\ell$ is a sample-wise loss function quantifying the error of the model with a weight $\w \in \mathbb{R}^{d}$ estimating the label for an input in $\z_{m,j}$. 
Suppose our global objective function is denoted as $f(\w)$, a federated learning problem may be formulated as 
\begin{equation}\label{global_objective_function}
    \min_{\w \in \mathbb{R}^{d}} f(\w) = \frac{1}{M} \sum_{m=1}^{M} f_{m}(\w).
\end{equation}
FedAvg is a popular federated learning method~\cite{mcmahan2017communication}. 
It adopts a periodic averaging strategy comprising three steps within each communication round. 
First, in communication round $k$, a global server will broadcast its weight vector $\w[k]$ to each worker. 
Second, each worker uses a gradient descent based method to independently update its own model for $\uptau$ local iterations. 
In particular, for local iteration $t\in [0,\uptau-1]$, given a mini-batch $\xi^{(k, t)}_{m} \subset \mathcal{D}_{m}$ of data points uniformly randomly
drawn from $\mathcal{D}_{m}$, the local weight $\w[k,t]_{m}$ may be updated as follows:
\begin{equation}\label{eq:local_update_rule}
    \w[k,t+1]_{m} = \w[k,t]_{m} - \eta \, \nabla f_{m}(\w[k,t]_{m}; \xi^{(k, t)}_{m}),   
\end{equation}
where $\eta$ is the learning rate. 
After $\uptau$ local iterations, each worker will obtain an updated local model $\w[k, \uptau]_{m}$.
In the final step, the server calculates the global weight by aggregating local weights $\w[k,\uptau]_{m}$'s from all workers, namely,
\begin{equation}\label{def:agg}
    \w[k+1] = \frac{1}{M} \sum_{m=1}^{M} \; \w[k,\uptau]_m.
\end{equation}
The algorithm then proceeds into the next communication round 
when the server broadcasts $\w[k+1]$ to each worker. 

\subsection{Compression Model}\label{section:communication_compression}
FedAvg assumes that the $m$th worker uploads the local weight/gradient vector to the server.
However, modern deep learning models tend to have a large number of weights, which can lead to a prohibitive communication cost. 
Instead of transmitting the original weights, one can use an encoding function $\enc(\cdot)$ to allow the worker to transmit a compressed version of the weight vector, 
\begin{equation}\label{eq:binary_code}
    \code[k]_m = \enc(\w[k,\uptau]_m; \mathcal{M}_m), 
\end{equation}
where $\mathcal{M}_m$ is some historical information in the memory of the $m$th worker.
For example, the broadcasted weight vectors in the previous communication rounds can be kept in the memory $\mathcal{M}_m$. 
The server will use a decoding function $\dec(\cdot)$ to reconstruct the local model 
\begin{equation}\label{eq:decode}
    \hw[k,\uptau]_m =  \dec(\code[k]_m; \mathcal{M}),
\end{equation}
where $\mathcal{M}$ is some historical information in the memory of the server.
The aggregation process in (\ref{def:agg}) will be modified to  
\begin{equation}\label{def:compressed_agg}
    \w[k+1] = \frac{1}{M} \sum_{m=1}^{M} \hw[k,\uptau]_m.
\end{equation}

\subsection{Transmission Model}\label{section:transmission_model}
In this work, we assume orthogonal frequency division multiple access (OFDMA) is employed to transmit the local updates to the server. 
The interference between different workers is ignored for simplicity.  
We use the channel capacity $c_m$ to estimate the uplink rate
\begin{equation}
    c_m = B \log_2 \left(1 + \frac{P_m h_m^2}{B N_0}\right), 
\end{equation}
where $B$ is the bandwidth, $P_m$ is the transmission power of worker $m$, 
$h_m$ is the channel gain, and $N_0$ is the noise power spectral density. 
We assume a quasi-static channel with channel gain   
\begin{equation}
    h_m = A_{d} \left( \frac{3\cdot 10^{8}}{4 \pi f_c d_m} \right)^{d_{e}}, 
\end{equation}
where $A_{d}$ is the antenna gain, $f_c$ is the carrier frequency, $d_m$ is the distance between worker $m$ and the server, and $d_e$ is the path loss exponent. 
For the downlink transmission, we assume the model will be broadcast through an error--free channel.  


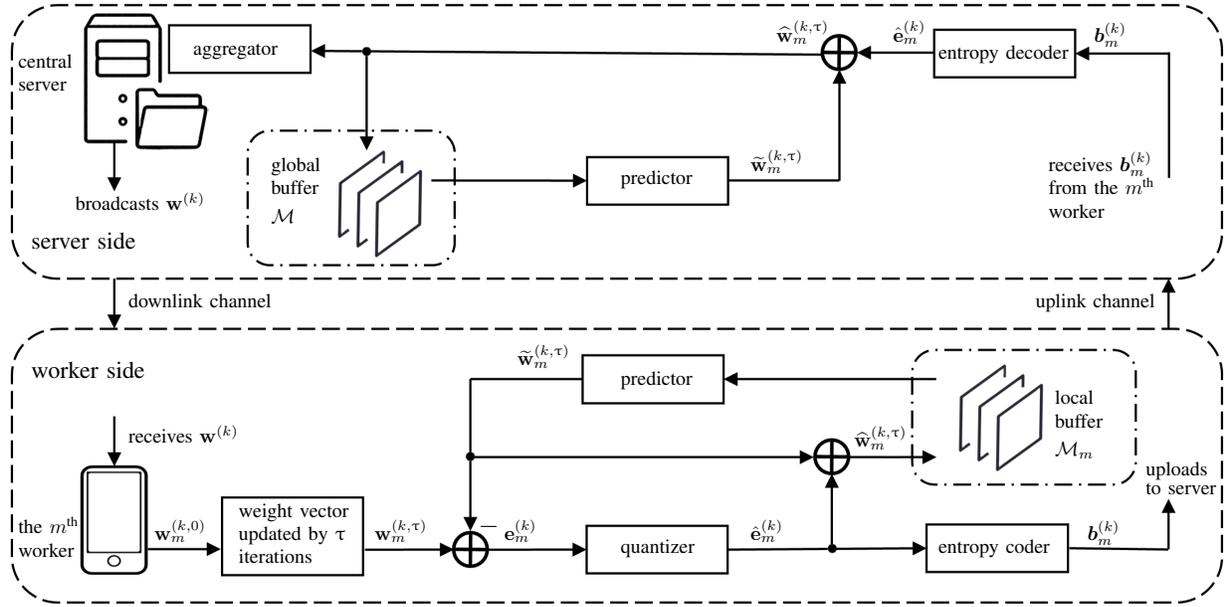
\begin{figure*}[!tb]
    \centering
    \begin{overpic}[width=0.95\textwidth]{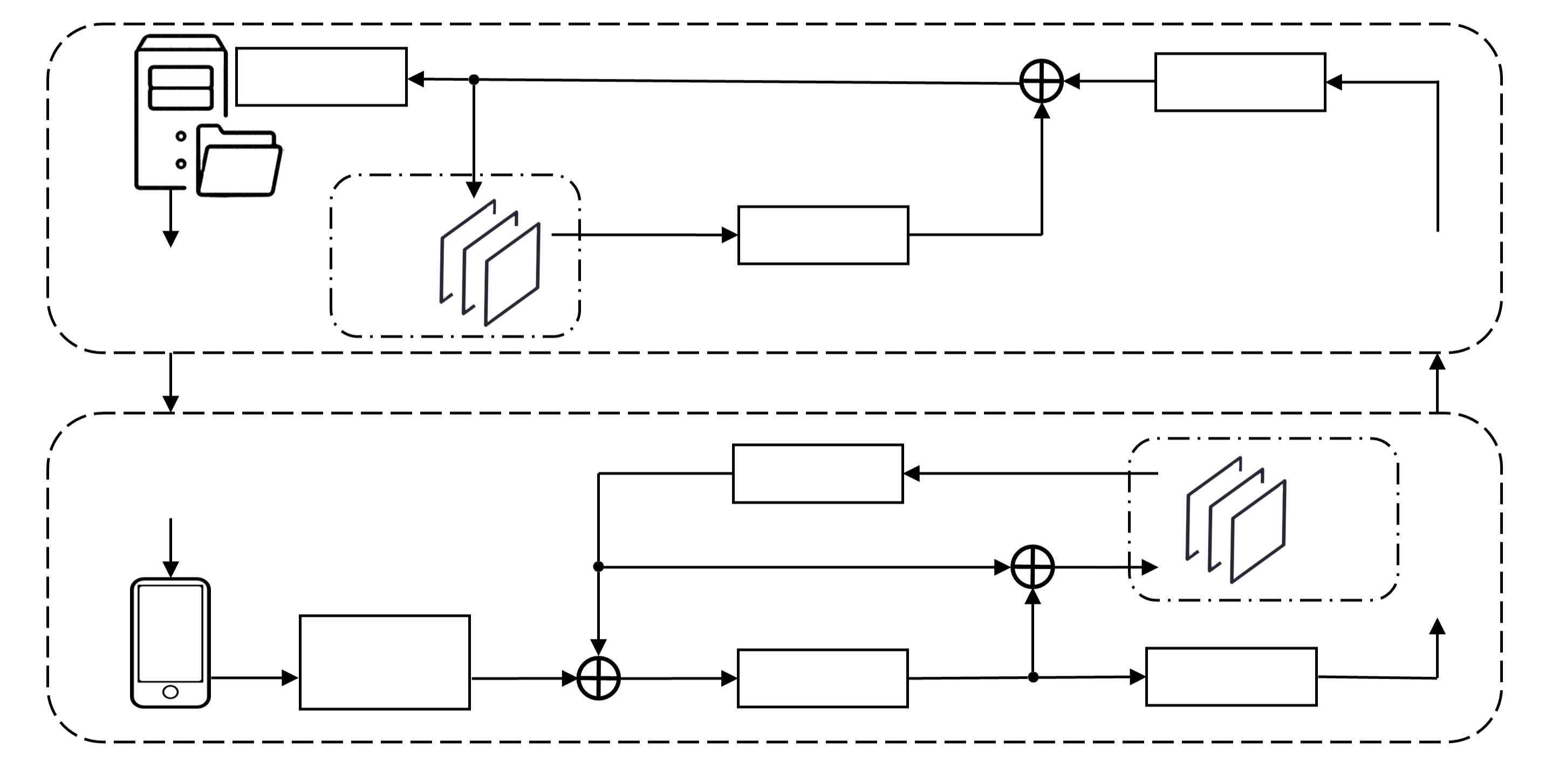}
    \put(12, 25.5){\intab{downlink channel}}
    \put(82, 25.5){\intab{uplink channel}}
    \put(4.5, 20){\intab[0.9]{worker side}}
    
    \put(3.5, 8){\intab{the $m^{\text{th}}$}}
    \put(3.5, 6.2){\intab{worker} }
    \put(12, 15){\intab{receives $\w[k]$}}

    \put(14, 7.5){\intab{$\w[k,0]_{m}$}}
    \put(20.5, 9.3){\intab{weight vector}} 
    \put(20.5, 7.5){\intab{updated by $\uptau$}} 
    \put(20.5, 5.7){\intab{iterations}}
    \put(31, 7.5){\intab{$\w[k,\uptau]_{m}$}}

    \put(39, 8){$-$}
    \put(50, 6.5){\intab{quantizer}}
    \put(41, 7.5){\intab{$\res[k]_{m}$}}
    \put(60, 7.5){\intab{$\hres[k]_{m}$}}
    \put(68, 14.7){\intab{$\hw[k, \uptau]_{m}$}}
    
    \put(83.5, 18){\intab{local} }
    \put(83.5, 16.2){\intab{buffer} }
    \put(83.5, 14){\intab{$\mathcal{M}_{m}$}}

    \put(42, 21){\intab{$\pw[k,\uptau]_{m}$}}
    \put(50, 19.5){\intab{predictor}}

    \put(74.5, 6.5){\intab{entropy coder}}
    \put(86, 7.5){\intab{$\code[k]_{m}$}}
    
    \put(90.5, 12.5){\intab{uploads}}
    \put(90.5, 11){\intab{to server}}
    
    \put(4.5, 30){\intab[0.9]{server side}}

    \put(3.5, 44){\intab{central}}
    \put(3.5, 42.2){\intab{server}}

    \put(7.9, 33){\intab{broadcasts $\w[k]$}}
    
    \put(17, 45){\intab{aggregator}}
    \put(23, 36){\intab{global}}
    \put(23, 34.2){\intab{buffer}}
    \put(23, 32){\intab{$\mathcal{M}$}}

    \put(50, 35){\intab{predictor}}
    \put(60, 36){\intab{$\pw[k,\uptau]_{m}$}}

    \put(74.5, 44.7){\intab{entropy decoder}}
    \put(71, 46){\intab{$\hres[k]_{m}$}}
    \put(62, 46){\intab{$\hw[k,\uptau]_{m}$}}
    \put(86.5, 46){\intab{$\code[k]_m$}}

    \put(83, 36){\intab{receives $\code[k]_{m}$}}
    \put(83, 34.2){\intab{from the $m^{\text{th}}$}}
    \put(83, 32.4){\intab{worker}}
    \end{overpic}
    \caption{Proposed federated learning scheme via predictive coding. 
    First, each worker receives $\w[k]$ from the server. 
    Second, each worker performs $\uptau$ local iterations to update the model.  
    The $m$th worker will obtain $\w[k,\uptau]_{m}$ after $\uptau$ iterations, and then use the predictor to estimate the weights based on its local buffer $\mathcal{M}_m$.  
    The residue between updated weights $\w[k,\uptau]_{m}$ and the predicted version $\pw[k,\uptau]_{m}$ will be fed into the quantizer. 
    The quantized residue $\hres[k]_{m}$ is then entropy coded and uploaded to the server.
    On the server side, the decoding procedure is performed as the inverse operation of the encoder. 
    }
    \label{fig:predfl}
\end{figure*}

\section{Proposed Predictive Coding-Based Compression Scheme}\label{coding_framework_section}

The proposed predictive coding-based compression scheme comprises three steps, namely, prediction, quantization, and entropy coding. 
As reviewed in Section~\ref{section:system_model}, for a worker $m$, instead of sending to the server the locally updated model $\w[k,\uptau]_m$ generated by iteratively invoking \eqref{eq:local_update_rule} for $\uptau$ local iterations, one can send the compressed version in \eqref{eq:binary_code} to lower the required bandwidth.
Our proposed scheme chooses to send to the server a compressed version of a residue $\res[k]_{m} = \w[k,\uptau]_{m} - \pw[k,\uptau]_{m}$.
Here, $\pw[k,\uptau]_{m}$ is a predicted version of $\w[k,\uptau]_{m}$ generated by combining the historical global weights broadcasted losslessly to workers. 
An ideal predictor will be able to decorrelate all coordinates of the residual vector $\res[k]_{m}$, boosting the effectiveness of the subsequent scalar quantization and entropy coding.
During the decoding process, the same prediction process of every worker is conducted to reconstruct the final local weight vectors $\{\hw[k, \uptau]_m\}_{m=1}^M$ and they are used for model aggregation as described by \eqref{def:compressed_agg}.
The procedure is illustrated in \paperfig{fig:predfl}, 
and the pseudocode is shown in Algorithm~\ref{coding_algorithm}. 
Below, we explain the proposed design in more detail.

\begin{algorithm}[tb]
    \caption{Predictive Coding Based Compression \label{coding_algorithm}}
    \SetKwComment{Comment}{$\triangleright$ }{}
    \SetKwBlock{server}{on server \textbf{do}}{}
    \SetKwBlock{worker}{for each worker $m = 1 : M$ \textbf{do}}{}
    \SetKw{init}{initialize}
    \SetKw{rec}{receive}
    \SetKw{pull}{pull}
    \SetKw{from}{from}
    \SetKw{push}{push}
    \SetKw{send}{send}
    \SetKw{to}{to}
    \SetKw{broadcast}{broadcast}
    \DontPrintSemicolon
    
    \init global weight $\w[0]$, local datasets $\{\mathcal{D}_{m}\}_{m=1}^{M}$ \\
    \For{$k=0,\, 1,\, \ldots, N $}{   
        \worker{
            \rec global weight $\w[k]$ \from server \\
            \init local weight $\w[k, 0]_m \leftarrow \w[k]$   \\
            \For{$t = 0 : \uptau-1$}{
                $\w[k,t+1]_{m} = \w[k,t]_{m} - \eta\, \nabla f_{m}(\w[k,t]_{m}; \xi^{(k,t)}_{m})$
            }
            $\pw[k, \uptau]_{m} = \pred(\w[k,0]_{m}; \mathcal{M}_{m})$ 
            $\res[k]_{m} \leftarrow \w[k,\uptau]_{m} - \pw[k,\uptau]_{m} $ \\ 
            $\{\sym^{(k)}_{m},\,\|\res[k]_m\|_{p} \} \leftarrow \quant(\res[k]_{m}) $ 
            $\code[k]_m \leftarrow \eenc(\sym^{(k)}_{m}) $  \\
            \send $\{ \code[k]_m, \|\res[k]_m\|_{p} \}$ \to server 
        }
    
        \server{
            \rec $\{\code[k]_m, \|\res[k]_m\|_{p}\}_{m=1}^M$ \from all workers   \;
            \For{$m = 1:M$}{   
                $\pw[k,\tau]_{m} = \pred(\w[k]; \mathcal{M})$
                $\sym^{(k)}_m \leftarrow \edec(\code[k]_m)$ \\ 
                $\hres[k]_{m} \leftarrow \dequant(\sym^{(k)}_m, \|\res[k]_{m}\|_{p})$ \\ 
                $\hw[k, \uptau]_{m} \leftarrow \pw[k,\uptau]_{m} +  \hres[k]_{m} $
            }
            $\w[k+1] = \frac{1}{M} \sum_{m=1}^{M} \hw[k, \uptau]_m$ \\ 
            \broadcast $\w[k+1]$ \to all workers
        }
    }
    \end{algorithm}

\subsection{Design of Predictor}\label{section:predictor}

To introduce the prediction schemes, we first examine model updates for each worker.
We define the accumulative local update $\de^{(k)}_{m}$ as the difference between the initial local weight vector $\w[k,0]_{m}$ and the final local updated weight vector $\w[k,\uptau]_{m}$ after $\uptau$ iterations, namely, 
\begin{equation}\label{eq:local_delta_def}
    \de^{(k)}_{m} \triangleq \w[k,0]_{m} - \w[k,\uptau]_{m}.
\end{equation}
With the update rule in (\ref{eq:local_update_rule}), $\de^{(k)}_{m}$ is equal to the accumulative local gradients scaled by the learning rate, i.e., 
\begin{equation}\label{eq:local_delta_accgrad}
    \de^{(k)}_{m} = \eta \sum_{t=0}^{\uptau-1} \nabla f_{m}(\w[k,t]_{m}; \xi^{(k,t)}_{m}). 
\end{equation}
Due to the use of the lossy compression scheme, the decoder on the server can only get access to the imperfectly reconstructed final local weight vector $\hw[k,\uptau]_{m}$. 
To ensure the consistency of the encoder and the decoder, the local memory will track $\hde^{(k)}_{m}$, which is defined as 
\begin{equation}\label{eq:local_hdelta_def}
    \hde^{(k)}_{m} \triangleq \w[k,0]_{m} - \hw[k,\uptau]_{m}.
\end{equation}

When the learning rate $\eta$ and number of local iterations are reasonably selected, the variance of $\de^{(k)}_{m,i}$'s will be smaller than that of $w^{(k,\uptau)}_{m,i}$'s. 
See \paperfig{fig:weight_and_gradient_distribution} for the empirical results. 
In other words, the weight updates have lower entropy compared with the original weights. 
Inspired by DPCM, we can compress the model difference vector $\de[k]_{m}$ instead of the weight vector $\w[k,\uptau]_{m}$. 
By doing so, we recover some recently proposed methods such as FedPAQ~\cite{reisizadeh2020fedpaq}. 
We list this method as the first prediction mode in \papertab{tab:prediction_modes}.

\begin{table}
    \caption{\sc Different Prediction Modes and Space Complexity \label{tab:prediction_modes}}
    \begin{tabular}{p{0.05\textwidth}p{0.2\textwidth}p{0.15\textwidth}}
        \toprule
        Mode & Prediction for $\w[k,\uptau]_{m}$ & Space Complexity \\ \midrule
        1 & $\w[k,0]_{m}$ & $d$ \\[5pt]
        2 & $\coeff[k] \circ \w[k,0]_{m} + \coeff[k]_{0}$ & $3\,d$ \\[5pt]
        3 & $\w[k,0]_{m} - \frac{1}{R} \sum\nolimits_{r=1}^R \hde^{(k-r)}_{m}$ & $(R+1)\, d$ \\[5pt]
        4 & $\w[k,0]_{m} - c \, \firmoment[k]_{m} \circ \denom[k]_{m}$ & $3\, d$ \\ \bottomrule 
    \end{tabular}
\end{table}   

\begin{figure}[tb]
    \centering
    \subcaptionbox{\label{subfig:accumulative_grad_distribution}}[0.24\textwidth]{
        \begin{overpic}[width=\linewidth]{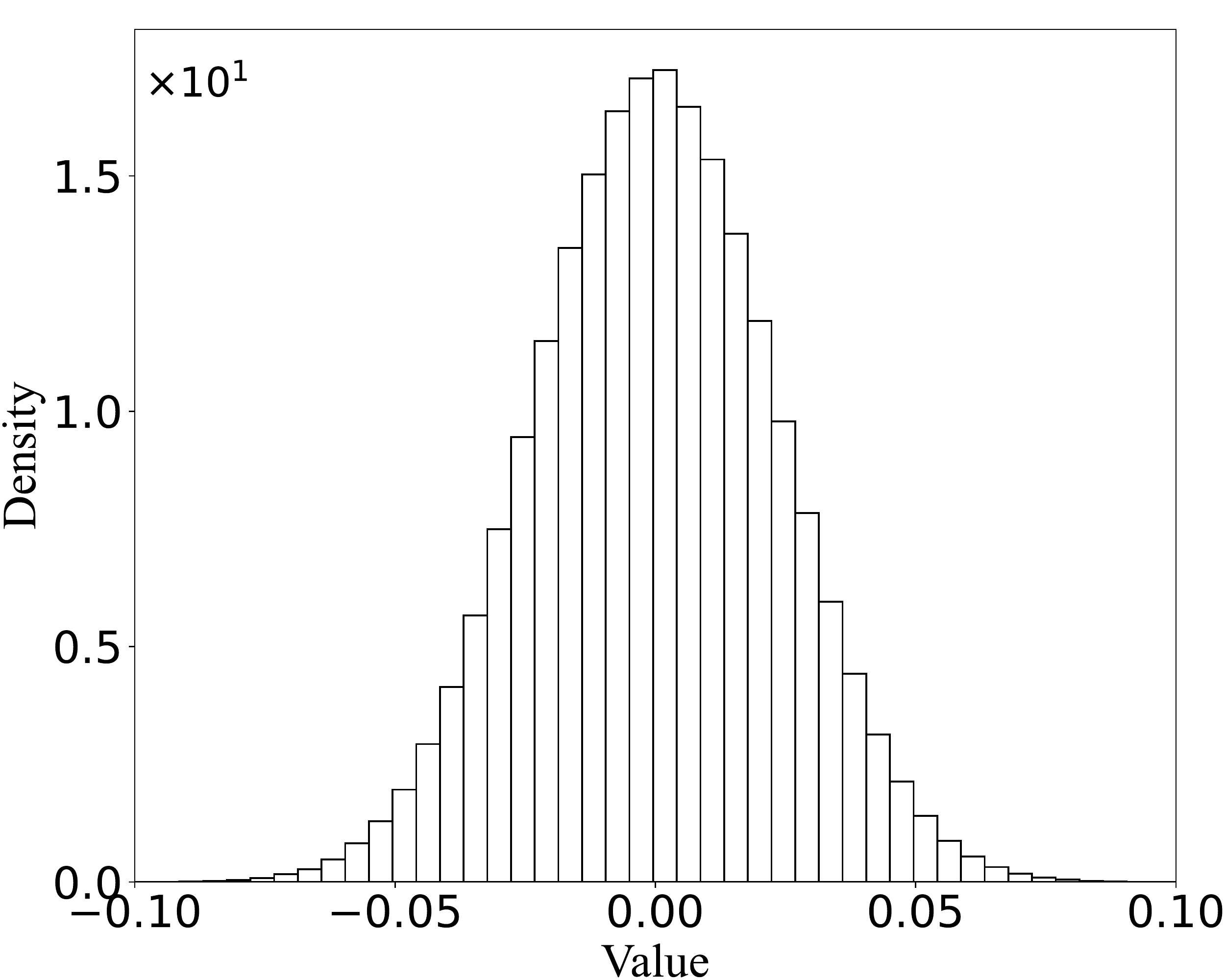}
        \put(13,60){\intab[0.55]{$\sigma^2 = 5.4 \times 10^{-4}$}}
        \end{overpic}
    }
    \subcaptionbox{\label{subfig:latent_prob_distribution}}[0.24\textwidth]{
        \begin{overpic}[width=\linewidth]{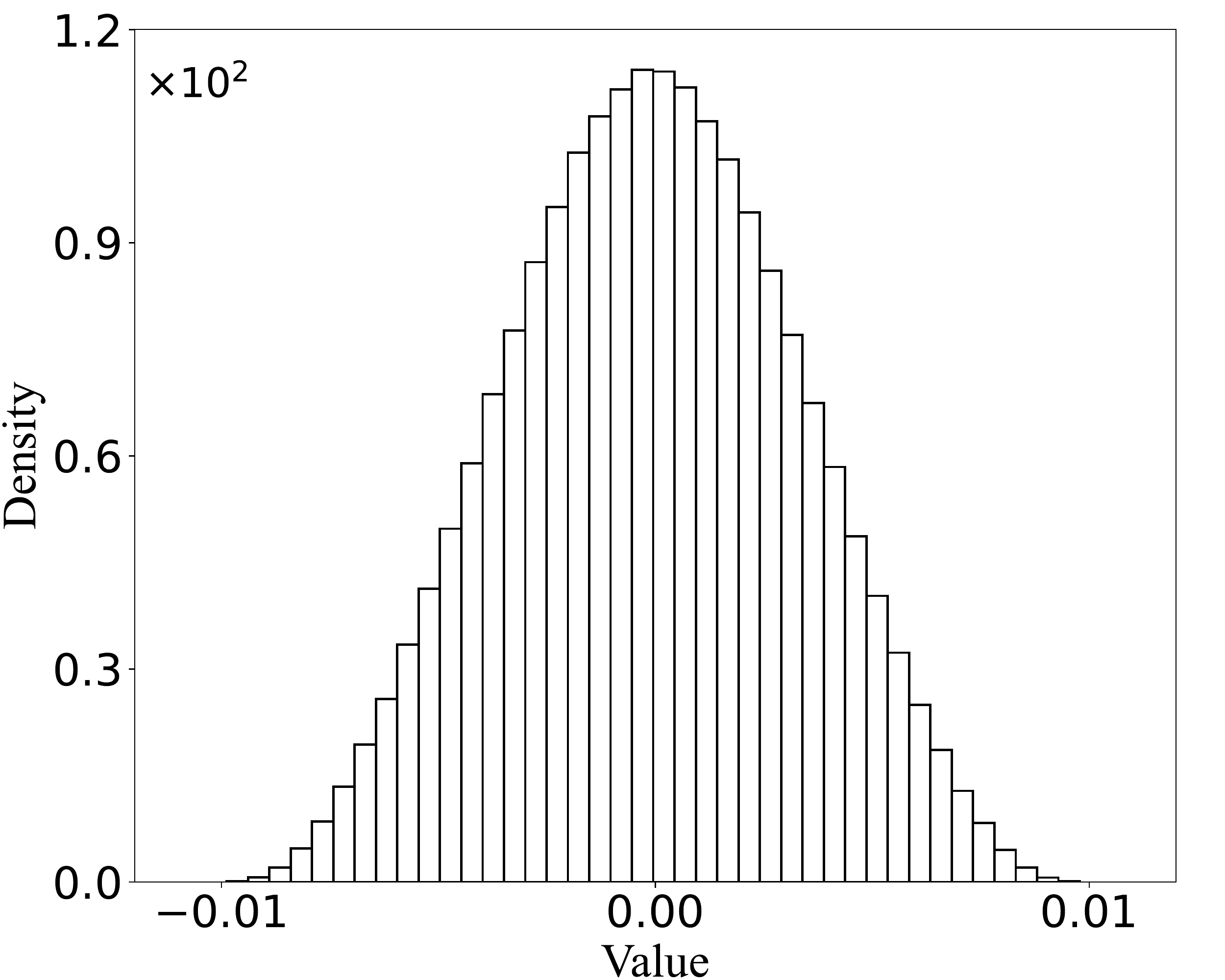}
            \put(13.,60){\intab[0.55]{$\sigma^2 = 1.1 \times 10^{-5}$}}
        \end{overpic}
    }
    \caption{
    Histograms of (a) weights $w^{(k,\uptau)}_{m,i}$ and 
    (b) model updates $\Delta^{(k)}_{m,i}$.
    The results are obtained from a single worker by training a convolutional neural network for $\uptau=20$ iterations.  
    Note that by using the model updates, the spread of information is reduced by an order of magnitude than using the original model weight. 
    This corresponds to a reduction of approximately $1$ bit per coordinate in entropy.
    }
    \label{fig:weight_and_gradient_distribution}
\end{figure}

The design of the prediction function can be also inspired from the properties of SGD. 
For example, Mandt et al.~\cite{mandt2016variational} use a multivariate Ornstein--Uhlenbeck process to approximate SGD: 
\begin{equation}\label{eq:ou_process}
    \diff \w[k,t+1]_{m} = -\textrm{A} \, \w[k,t]_{m} \diff t + \textrm{B} \diff \mathbf{W}_{t}, 
\end{equation} 
where $\textrm{A},\, \textrm{B} \in \mathbb{R}^{d\times d}$ are constant matrices, and $\mathbf{W}_{t}$ is a $d$-dimensional Wiener process. 
Based on the solution to \eqref{eq:ou_process}, an autoregressive (AR) model can be leveraged to design the predictor. 
We use a coordinate-wise linear predictor by introducing coefficients $\coeff[k]_{m} \in \mathbb{R}^{d}$ and biases $\coeff[k]_{0,m} \in \mathbb{R}^{d}$, i.e., 
\begin{equation}
    \pw[k,\uptau]_{m} = \coeff[k]_{m} \circ \w[k,0]_{m} + \coeff[k]_{0,m},
\end{equation}
where $\circ$ is the elementwise product. 
Similar to the adaptive DPCM method, we use the gradient descent to update coefficients. 
We first calculate the mean squared error (MSE) of the prediction, i.e., 
\begin{equation}
    \mse[k]_{m} = \frac{1}{d} \left\| \pw[k,\uptau]_{m} - \hw[k,\uptau]_{m} \right\|_2^2,  
\end{equation}
and update the coefficients with step size $a$,
\begin{equation}
    \coeff^{'(k)}_{m} = \coeff[k]_{m} - a \, \frac{\partial\, \mse[k]_{m}}{\partial\, \coeff[k]_{m}}.    
\end{equation}
Here, we stress that the MSE measures the difference between the predicted version $\pw[k,\uptau]_{m}$ and the reconstructed version $\hw[k,\uptau]_{m}$. 
Note that by transmitting the quantized residue $\hres[k]_{m}$, $\hw[k,\uptau]_{m}$ is available at both the worker side and the server side to allow the prediction coefficients to be synchronized.
We list this method as the second prediction mode in \papertab{tab:prediction_modes}.

Another way to generate a prediction for $\w[k,\uptau]_{m}$ is to exploit the relation in \eqref{eq:local_delta_def} to obtain an estimated $\pde[k]_{m}$. 
The predicted $\pw[k,\uptau]_{m}$ can be written as 
\begin{equation}
    \pw[k,\uptau]_{m} = \w[k,0]_{m} - \pde[k]_{m}.  
\end{equation}
By assuming the recent increments are correlated, 
we use the moving average of the globally available weight differences to do the prediction, namely, 
\begin{equation}
    \pde[k]_{m} = \frac{1}{R} \sum_{r=1}^{R} \hde^{(k-r)}_{m}. 
\end{equation}
We list this method as the third prediction mode in \papertab{tab:prediction_modes}.

Equation \eqref{eq:local_delta_accgrad} illustrates that the weight difference $\de^{(k)}_{m}$ is the accumulative local gradient updates. 
We borrow the wisdom from the adaptive moment estimation (Adam) optimizer~\cite{kingma2015adam} by smoothing the trajectory of $\hde^{(r-k)}_{m}$ for prediction. 
In particular, we take the exponential moving average to track the first and second raw moments, i.e., 
\begin{subequations}
\begin{align}
    \firmoment[k]_{m} &= \beta_1 \, \firmoment[k-1]_{m} + (1-\beta_1)\, \hde[k-1]_{m}, \\
    \secmoment[k]_{m} &= \beta_2 \, \secmoment[k-1]_{m} + (1-\beta_2)\, \hde[k-1]_{m} \circ  \hde[k-1]_{m},
\end{align}
\end{subequations}
where $\beta_{1}$ and $\beta_{2}$ are predefined scalar coefficients. 
The prediction is designed as 
\begin{equation}
    \pde[k]_{m,i} = c \, \frac{\firmom[k]_{m,i}}{\sqrt{\secmom[k]_{m,i} } + \varepsilon}, 
\end{equation}
where $c$ is a constant, and $\varepsilon$ is a small value added to the denominator for numerical stability. 
We simplify the notation by introducing a vector $\denom[k]_{m}$ such that 
\begin{equation}
    \pde[k]_{m} = c \, \firmoment[k]_{m} \circ \denom[k]_{m}. 
\end{equation}
We list this method as the fourth prediction mode in \papertab{tab:prediction_modes}.
To achieve the best compression performance, we traverse all prediction modes and choose the one that gives the smallest prediction error. 
Two-bit mode information is transmitted separately for signaling the selected mode to the decoder. 
The time complexity of the proposed prediction step involves standard matrix calculations and grows linearly with the number of modes $N$. 
We note that communication is the bottleneck in federated learning~\cite{reisizadeh2020fedpaq}. 
The local computation is dominated by neural network optimization, and the additional computational cost is therefore negligible.

The design of different prediction modes and their empirical performance have been well exploited in image and video compression~\cite{sze2014high, sayood2017introduction}.  
We provide a mathematical justification that the prediction error decreases as more predictor candidates are included. 
\begin{Lemma}\label{lemma:multimode_prediction}
    In one communication round, suppose we have a sequence of infinitely many prediction modes $i = 1, 2, \dots, $ each mode can result in an independent, nonnegative prediction error 
    that admits a probability density function $f_{X_i}(x)$. For the first $N$ candidate modes, a mode selection scheme picks the one with the lowest absolute error. 
    Denote the error as $Y_N$.
    The expectation of the smallest prediction error $Y_N$ is a monotonically decreasing function of $N$.
\end{Lemma}
\begin{proof}
    See \paperApp{proof:multimode_prediction}. 
\end{proof}

Finally, we discuss the issue of memory cost in the implementation.
In practice, the total number of workers $M$ can be a large value. 
Keeping the memory in sync with all workers on the server may incur additional computation and storage cost. 
With the globally available model weights, we can use 
\begin{equation}
    \hde^{\prime(k)}_{m} \triangleq \w[k,0]_{m} - \w[k+1] = \w[k] - \w[k+1]
\end{equation}
to replace the original $\hde[k]_{m}$ defined in \eqref{eq:local_hdelta_def}. 
This strategy reduces the memory cost from $\order{Md}$ to $\order{d}$ at the expense of less precise prediction for each worker. 

\subsection{Quantization}\label{quantization_subsection}
Quantization maps continuous input values to discrete symbols. 
In our work, we use the operator $\quant(\cdot)$ to represent the forward quantization stage. 
In particular, it decomposes an input vector into its norm and direction and quantizes the projected directions to discrete symbols.  
The operator $\dequant(\cdot)$ denotes the dequantization stage that reconstructs the input vector by reassembling the discrete symbols back to a directional vector and multiplying it by the norm.
The quantizer $Q(\cdot)$ is defined as the composition of the forward quantization $\quant(\cdot)$ and the dequantization $\dequant(\cdot)$. 
Given an input residue $\res$, the quantizer maps each entry as follows:
\begin{equation}
    Q(e_i) = \frac{\kappa}{s} \, \|\res\|_p \cdot \sign (e_i) \cdot \varphi_i(\res, s),\label{eq:quant_operator}
\end{equation}
where $\kappa$ is a scaling factor, $s$ is a predefined parameter and the number $L$ of representation levels is equal to $2s+1$, $\|\res\|_p$ is the $\ell_{p}$ norm of $\res$, and $\varphi_i(\res, s)$ is an integer value representing the unsigned quantized level of the $i$th coordinate of vector $\res$.

In this work, we use two types of quantizers.
For a deterministic mid-tread uniform quantizer $\uniformquant$ that has a zero-valued reconstruction level~\cite{schuller2020quantization}, we set $\varphi_i$ in \eqref{eq:quant_operator} by $\univar_i$ defined as follows:
\begin{equation}
    \univar_i(\res, s) = \left\lfloor \frac{s |e_i|}{ \kappa \|\res\|_p } + \frac{1}{2} \right\rfloor,
\end{equation}
whose effective quantization step is $\kappa \|\res\|_p/s$.
For a stochastic quantizer $\stoquant$, we set $\varphi_i$ in \eqref{eq:quant_operator} by $\stovar_i$ defined as follows:
\begin{equation}\label{def:normalization_qsgd}
    \stovar_i(\res, s) = \left\{
        \begin{array}{l @{\; \;} l}
        \ell &  \text { with prob. } 1-p_{i}, \\[5pt]
        \ell + 1 & \text { with prob. } p_{i} = \frac{s|e_i|}{\kappa \|\res\|_p} - \ell,
    \end{array}\right.
\end{equation}
where $\ell \in [0,s)$ is an integer and $\frac{|e_{i}|}{\kappa \|\res\|_{p}} \in[\ell / s,(\ell+1) / s]$~\cite{alistarh2017qsgd}. 
The unsigned quantized level given by $\stovar_i(\res, s)$ in its binary representation will be combined with the sign of the quantizer input $e_i$ and further compressed by an entropy coder. 
Specifically, the sign will be concatenated to the least significant bit to produce an unsigned quantized level $h_i$ as follows:
\begin{equation}\label{eq:quant_map}
    h_i = \left\{
    \begin{array}{l @{\;,\;} l}
        2\,  \varphi_i(\res, s)  & \text{ if } \sign(e_i) \leqslant 0, \\
        2 \,  \varphi_i(\res, s) - 1 & \text{ if } \sign(e_i) > 0.
    \end{array}
    \right.
\end{equation}
Formally, we show that the mapping scheme in \eqref{eq:quant_map} can reduce the average codeword length compared to the method that separately  encodes the signs and absolute integer values~\cite{alistarh2017qsgd}.

\begin{Lemma}\label{lemma:sign_mapping}
    Define $\phi_i \triangleq \sign(e_i) \cdot \varphi_i(\res, s). $
    The mapping defined in \eqref{eq:quant_map} can be rewritten as 
    \begin{equation}\label{eq:re_quant_map}
        h_i = \left\{
        \begin{array}{l @{\;,\;} l}
            -2\, \phi_i   & \textrm{ if } \sign(\phi_i) \leqslant 0, \\
            2 \,  \phi_i - 1 & \textrm{ if } \sign(\phi_i) > 0.
        \end{array}
        \right.
    \end{equation}
    Consider following nonnegative integer encoding schemes: 
    (i) use the mapping defined in \eqref{eq:re_quant_map} and then encode $\phi_i$; 
    (ii) separately encode $|\phi_i|$ and $\sign (\phi_i)$.
    Scheme (i) has a shorter average codeword length.      
\end{Lemma}
\begin{proof}
    See \paperApp{proof:lemma:sign_mapping}. 
\end{proof}

After the mapping, $\sym = [h_1, \ldots, h_{d}]^\top$ contains nonnegative integers that can be further compressed with an entropy coder.
Given the quantization error $D$ and the estimated entropy $R$, the quantizer is chosen by minimizing the Lagrangian cost function~\cite{ortega1998rate}, namely, 
\begin{equation}
    \mathcal{L} = D + \lambda \cdot R,
\end{equation}
where $\lambda$ is the Lagrangian multiplier. 
Note that the dequantization on the decoder side is the same for both quantizers $\stoquant$ and $\uniformquant$, hence there is no need to send the overhead indicating the choice of the quantizer.   
Other information such as the vector norm will be directly transmitted without compression.

\subsection{Entropy Coding}
Given a sequence of discrete symbols, the task of entropy coding is to find a mapping such that the inputs are represented with codewords that have a shorter weighted average length. 
We choose the \textit{arithmetic coding}~\cite{sayood2017introduction} in our implementation. 
Instead of mapping each symbol to a code uniquely, arithmetic coding encodes a tag that represents the cumulative probability distribution of an input random sequence. 
In practice, the tag will be set as a value located in an interval,  which is bounded by the cumulative probability distribution of the sequence.  
During the coding procedure, the algorithm refers to a predefined probability table and maintains the tag. 
The tag will be output as a binary bitstream.

In our scheme, we use the frequency of the quantized residues as an estimation of the probability table. 
The frequency information will be transmitted separately to ensure that the global server can decode the original sequence.

\begin{figure*}
\centering
\subcaptionbox{\label{fig:fmnist_loss_cost}}[0.32\textwidth]{
    \includegraphics[width=\linewidth]{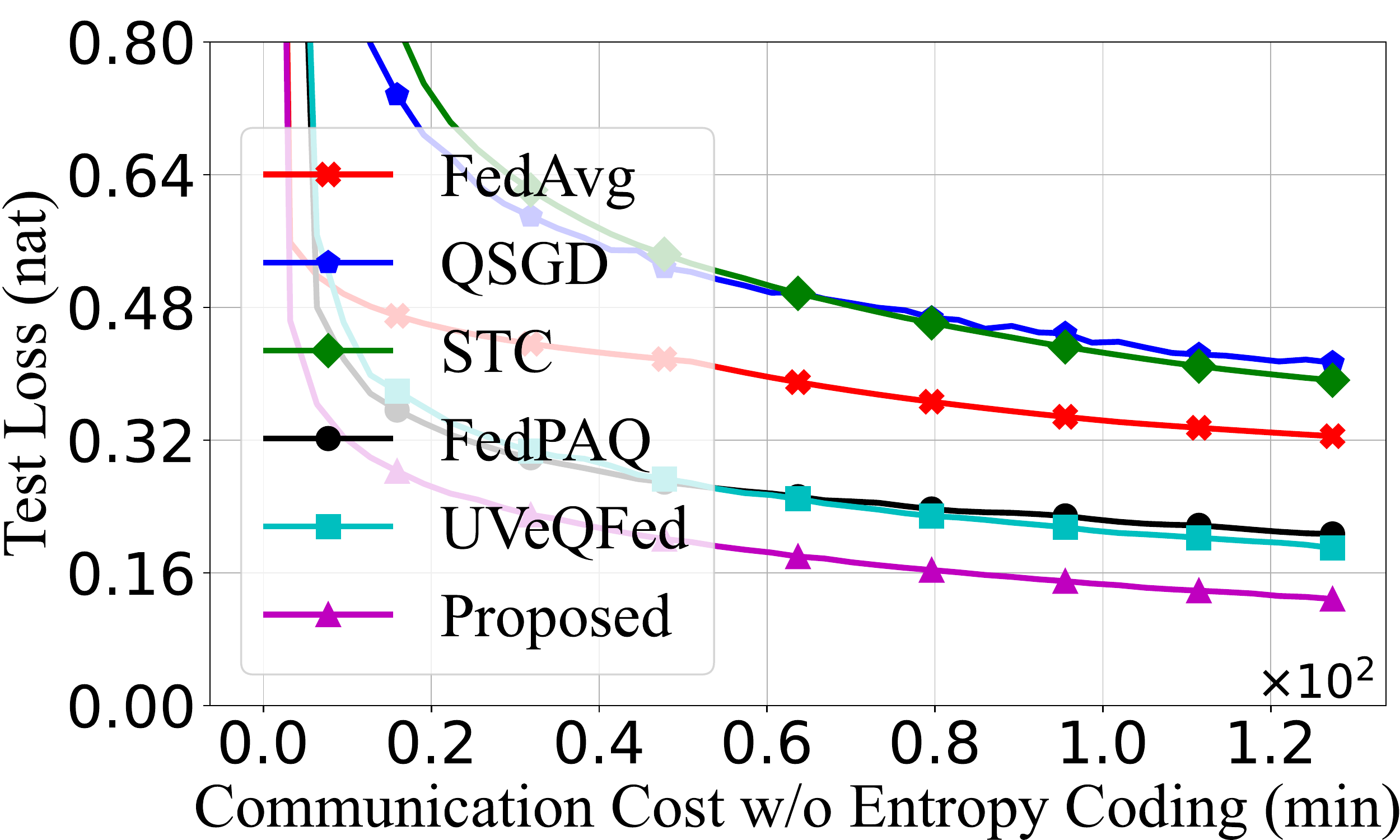}
}
\subcaptionbox{\label{fig:cifar_loss_cost}}[0.32\textwidth]{
    \includegraphics[width=\linewidth]{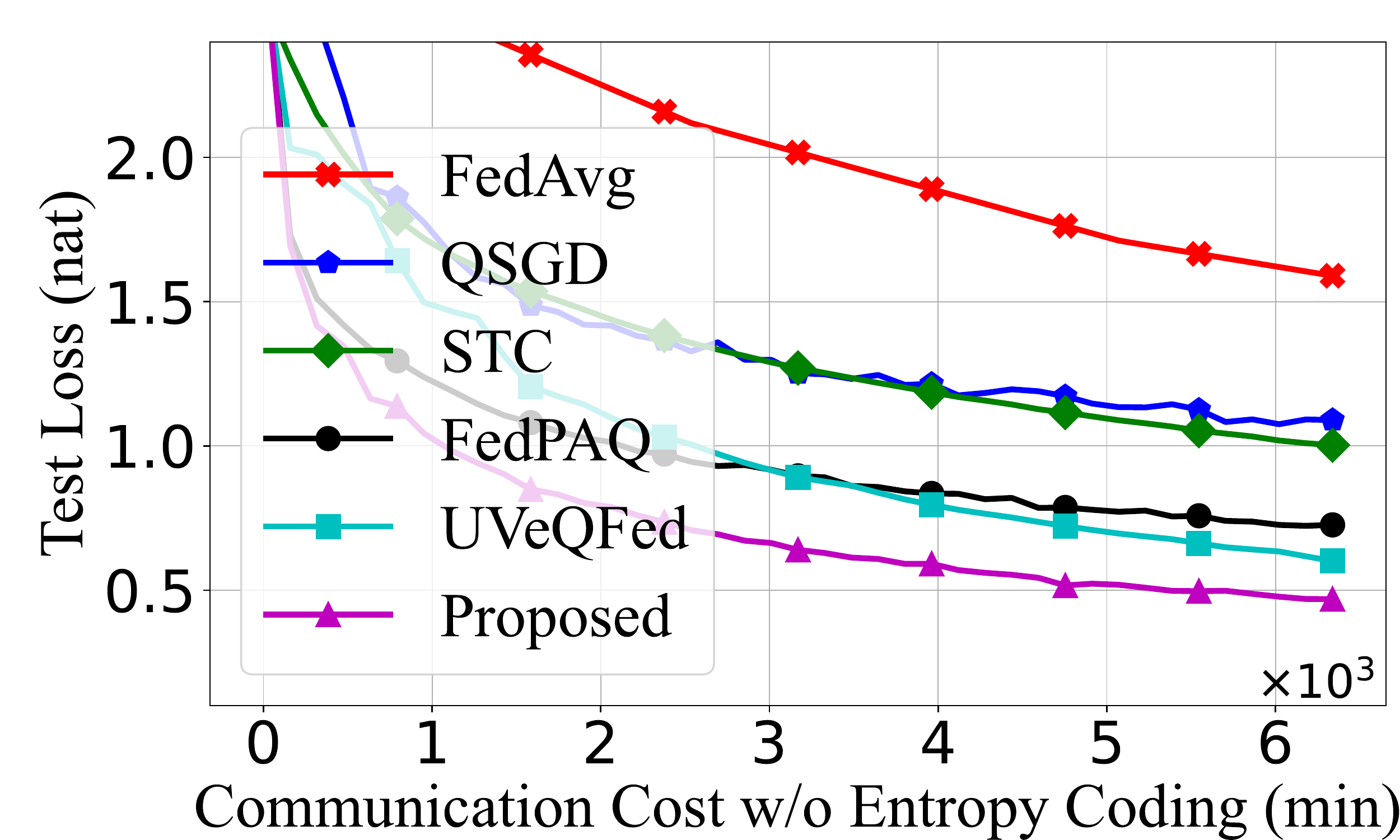}
}
\subcaptionbox{\label{fig:comm_time}}[0.32\textwidth]{
    \includegraphics[width=\linewidth]{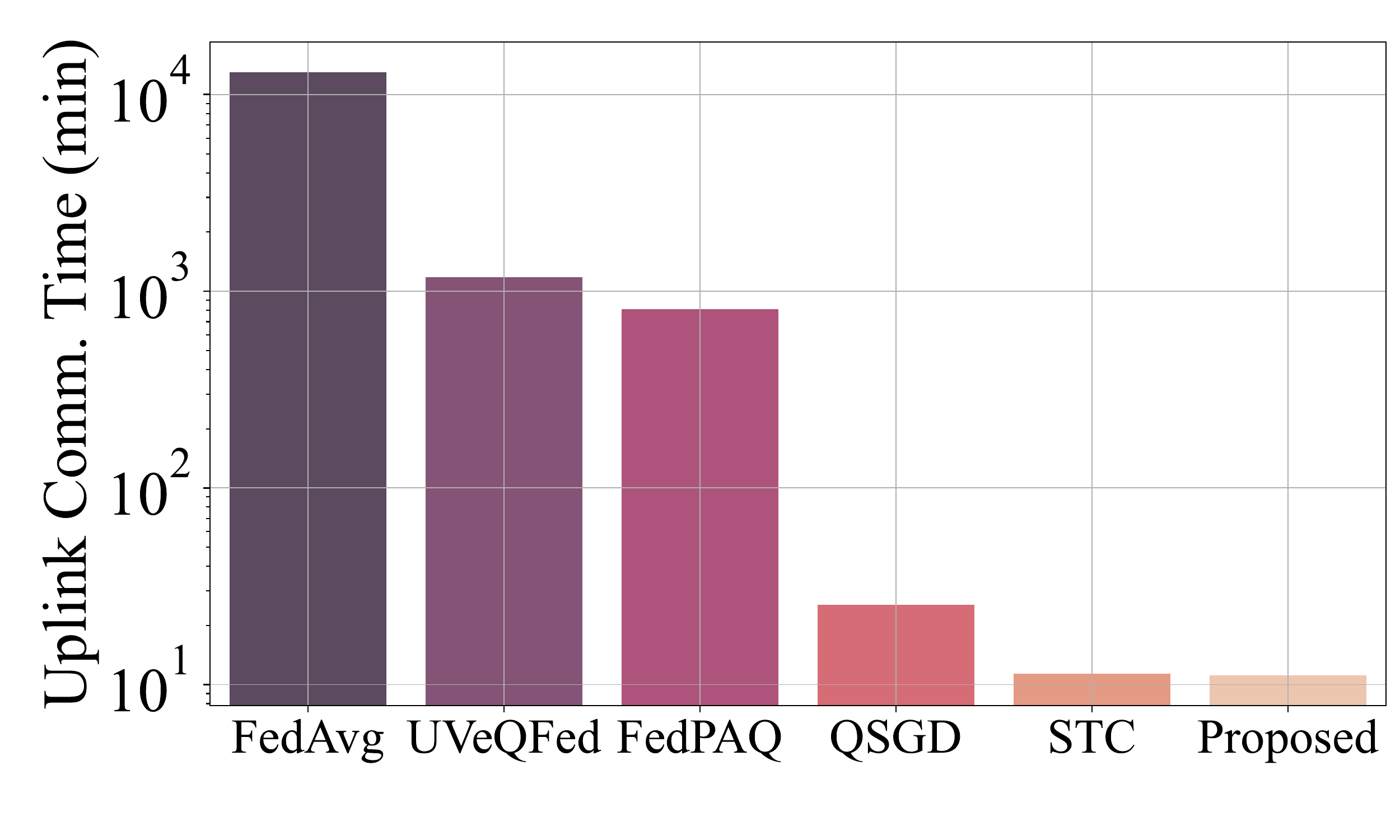}
}
\caption{
Test loss versus communication cost without entropy coding on 
(a)~the non-i.i.d. Fashion-MNIST dataset and 
(b)~the non-i.i.d. CIFAR-10 dataset.
The proposed scheme outperforms other methods by achieving the lowest test loss within the same uplink cost. 
(c)~Uplink communication time to reach $60\%$ test accuracy on the CIFAR-10 task.
The proposed scheme is advantageous in bandwidth-constrained scenarios by saving communication time drastically.} 
\label{fig:comparison}
\end{figure*}

\section{Analysis of Algorithm}\label{section:algorithm_analysis}
To simplify the notation, we first denote the stochastic local gradient described in (\ref{eq:local_update_rule}) as 
\begin{equation}
    \stograd[k,t] \triangleq \nabla f_m(\w[k,t]_{m}; \xi_{m}^{(k,t)}). 
\end{equation}
In addition, the local full batch gradient will be shortened as
\begin{equation}
    \grad[k,t] \triangleq  \nabla f_m(\w[k,t]_{m}).
\end{equation}
We state five assumptions as prerequisites for the convergence analysis. 
We assume an optimization procedure with a fixed learning rate $\eta$ for mathematical tractability. 

\subsection{Assumptions}

\begin{Assumption}\label{assumption:lower_bound}
    (Lower bound) $\forall \; \w \in \mathbb{R}^{d}$, the objective function is 
    lower bounded by a constant $f^*$
    \begin{equation}
        f(\w) \geqslant f^* = \min_{\w \in \mathbb{R}^{d}} f(\w).
    \end{equation}
\end{Assumption}
\begin{Assumption}\label{assumption:Lsmooth}
    ($L$-smoothness) \hspace*{2pt} $\forall \; \w_1, \w_2 \in \mathbb{R}^d$, $m \in [1,\,M]$,  
    there exists some nonnegative $L$ such that:
    \begin{equation}
        \|\nabla f_{m}(\w_1) - \nabla f_{m}(\w_2) \|_2 \leqslant L \, \|\w_1 - \w_2 \|_2.
    \end{equation}
\end{Assumption}
\begin{Assumption}\label{assumption:sto_grad}
    The stochastic gradients on each worker are unbiased, namely, $\mathbb{E}_{\xi} [\stograd[k,t]] = \grad[k,t]$. 
    $\forall \; k,\, t > 0$, $m \in [1,M]$, they have bounded variance, i.e., 
    \begin{equation}
        \expect \normsq{\stograd[k,t] - \grad[k,t]}  \leqslant \sigma_{\xi}^{2}, \label{sample_noise_boundvar}
    \end{equation}
    where $\sigma_{\xi}^2$ is a fixed variance independent of $k$, $t$, and $m$. 
\end{Assumption}
Assumptions \ref{assumption:lower_bound} to \ref{assumption:sto_grad} are common in literature for necessary analyses~\cite{wang2018cooperative,haddadpour2020fedcom}. 
Next, we set conditions for the prediction  residue and quantization noise.

\begin{Assumption}\label{assumption:prediction}
    $\forall\; k >0$, $m \in [1,M]$, each component of the prediction residual vector $\res[k]_{m}$ has a distribution that is symmetric about zero. 
    In addition, the variance satisfies 
    \begin{equation}
        \expect \normsq{\res[k]_{m}} \leqslant p\, \normsq{\de[k]_{m}}, 
    \end{equation}
    for some constant $p > 0$.
\end{Assumption}
Assumption \ref{assumption:prediction} is a nonstandard assumption, and we provide some intuitions as follows.
Consider the prediction mode~$1$ inspired by DPCM. 
We expect the central limit theorem to render a bell-shaped distributed residue with zero skewness.
The equivalent assumption can be found in~\cite{bernstein2018signsgd}.  
For the variance bound of $\res[k]_{m}$, we have $p=1$ according to $\res[k]_{m} = -\de[k]_{m}$ in prediction mode $1$. 
Since the predictor is selected by minimizing the prediction error, we expect the prediction error ratio $p$ to be some value between $0$ and $1$.

\begin{Assumption}\label{assumption:quantization}
    $\forall\; \vecA \in \mathbb{R}^{d}$, the variance of quantization noise satisfies
    \begin{equation}
        \expect \normsq{ Q(\vecA) - \vecA } \leqslant q\, \normsq{\vecA},
    \end{equation}
    for some constant $q > 0$.
\end{Assumption}
Assumption \ref{assumption:quantization} essentially gives a lower bound of the signal to quantization noise ratio (SQNR) of the quantization. 
The lower bound can be shown for the deterministic quantizer $\uniformquant$ when assuming a specific distribution of the input~\cite{sayood2017introduction}, 
and for the stochastic quantizer $\stoquant$~\cite{alistarh2017qsgd}.

\subsection{Convergence Analysis}
We state our analysis result on the algorithm for nonconvex optimization in the i.i.d. setting when workers have the same data distribution.  
We use the gradient norm as the indicator for convergence~\cite{wang2018cooperative}, which is a necessary condition for achieving a local minimum.
An $\epsilon$-optimal solution is achieved when the average of squared gradient norm is bounded by $\epsilon$. 

\begin{Theorem}\label{theorem:convergence}
    For Algorithm \ref{coding_algorithm}, under Assumptions \ref{assumption:lower_bound} to \ref{assumption:quantization}, if the learning rate $\eta$ satisfies
    \begin{equation}\label{eq:constraint_lr}
        \frac{L^{2} \eta^{2} \uptau(\uptau-1)}{2} + L\eta\uptau \left(\frac{qp}{M} + 1\right) \leqslant 1,
    \end{equation}
    then after $K$ rounds of communication, we have 
    \begin{multline}\label{eq:convergence_bound}
        \frac{1}{K} \sum_{k=0}^{K-1} \expect \normsq{\nabla f(\w[k])} \leqslant 
        \frac{2\left[ f\left(\w[0]\right)-f\left(\w^*\right) \right]}{\eta \uptau K} \\
        + L \eta\left(\frac{qp+1}{M} + \frac{L \eta(\uptau-1)}{2}\right) \sigma_{\xi}^{2}.  
    \end{multline}
\end{Theorem}
\begin{proof}
See \paperApp{appendix:proof_convergence}.
\end{proof}

\begin{Remark}
Suppose we have an ideal predictor with the prediction error ratio coefficient $p$ close to zero, or equivalently we have an ideal  quantizer with the quantization error ratio coefficient $q$ close to zero, then Theorem \ref{theorem:convergence} recovers the result obtained in~\cite{wang2018cooperative}.
\end{Remark}
\begin{Remark}\label{remark:predict_quant}
In our proposed scheme, prediction and quantization operation have an interactive effect on the convergence rate, as it is reflected in the term $qp \frac{L\eta}{M} \sigma^2_{\xi}$ in~\eqref{eq:convergence_bound}.  
Suppose the prediction error ratio satisfies $p \in (0,1)$, then the convergence bound in \eqref{eq:convergence_bound} is tighter than gradient quantization methods such as FedPAQ~\cite{reisizadeh2020fedpaq}. 
The well-designed predictors allow fast convergence even in the coarse quantization scenario. 
\end{Remark}
\begin{Remark}\label{remark:eta_tau}
The learning rate $\eta$ or the number of local iterations $\uptau$ should not be large. 
Otherwise, the prediction will not be accurate. 
To see this, we replace $q$ with $d^{1/2}$ and replace $p$ with $p + \delta_{p}$. 
This brings an additive term $\frac{L\eta }{M} d^{\frac{1}{2}}\delta_{p}\, \sigma^2_{\xi}$ to the upper bound in \eqref{eq:convergence_bound},
thus negatively affecting the convergence and the prediction in succeeding communication rounds. 
\end{Remark}

Adding the uniform bounds for the gradients or gradient dissimilarity provides an analysis tool for non-i.i.d. settings~\cite{li2020on, karimireddy2020scaffold}, which we leave for future work. 
The simulation results with the heterogeneous data distributions are discussed in Section \ref{section:experiment}.   

\section{Experimental Results}\label{section:experiment}

\begin{figure*}
\centering
\subcaptionbox{\label{subfig:loss_modes}}[0.325\linewidth]{
    \includegraphics[width=\linewidth]{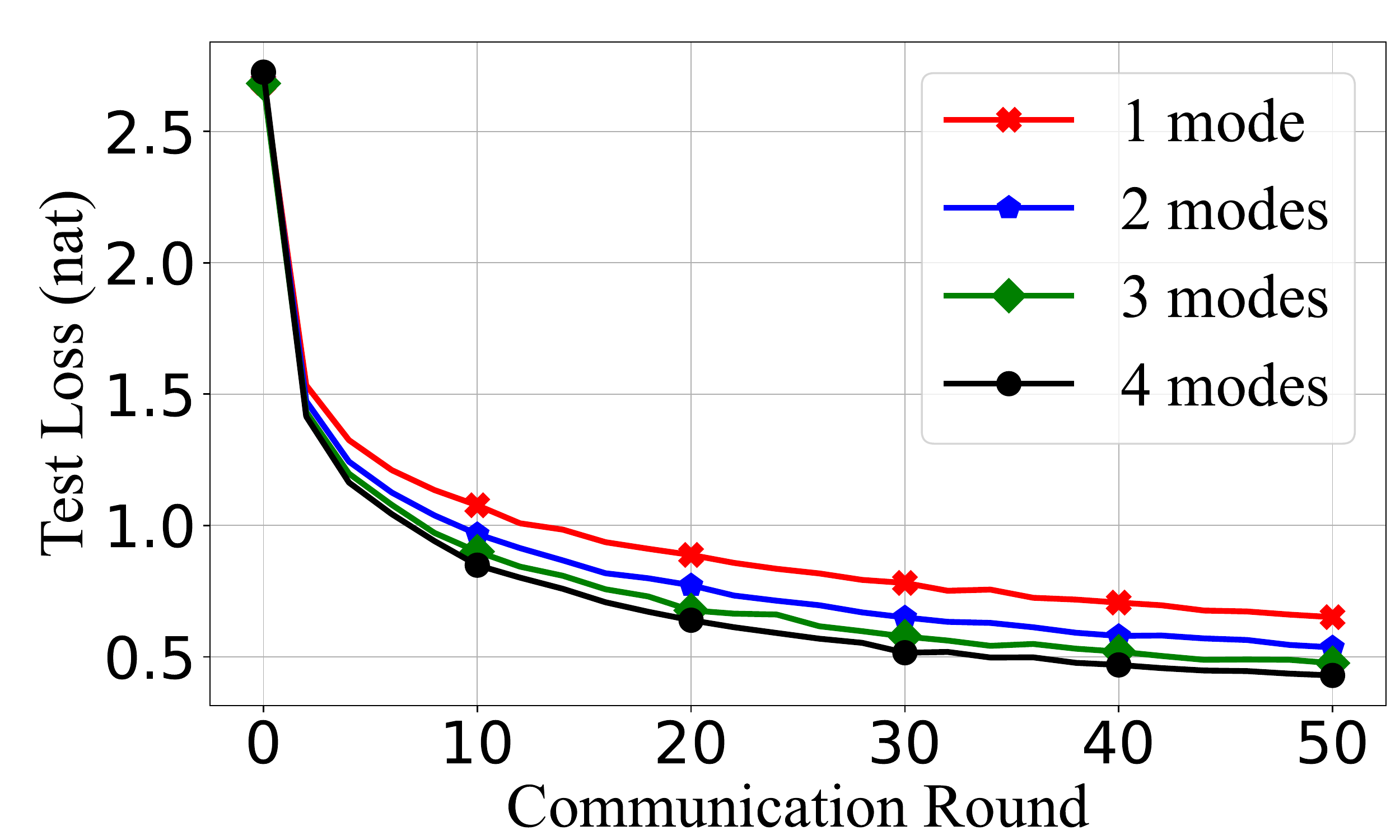}
}
\subcaptionbox{\label{subfig:freq_modes}}[0.325\linewidth]{
    \includegraphics[width=\linewidth]{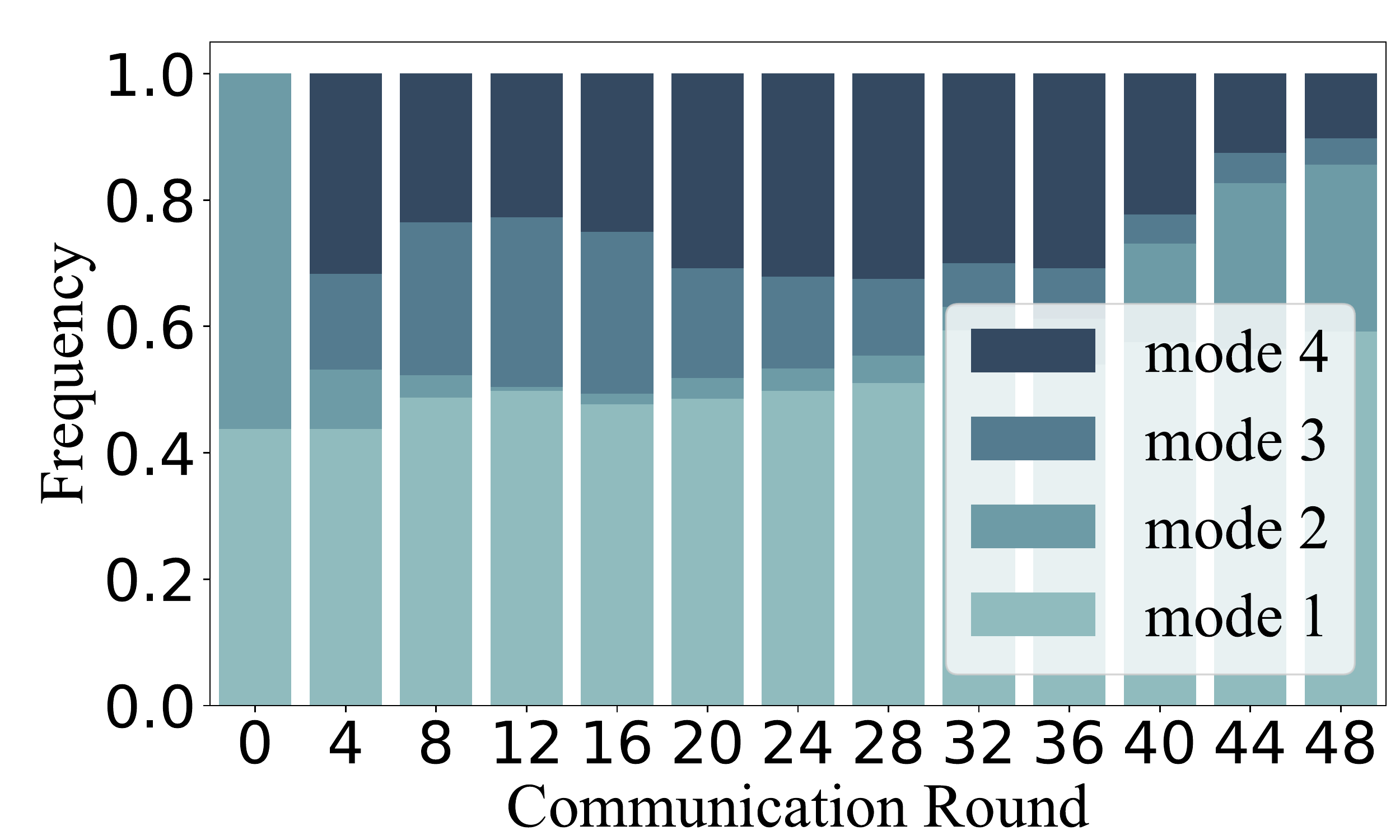}
}
\subcaptionbox{\label{fig:effect_tau}}[0.325\linewidth]{
    \includegraphics[width=\linewidth]{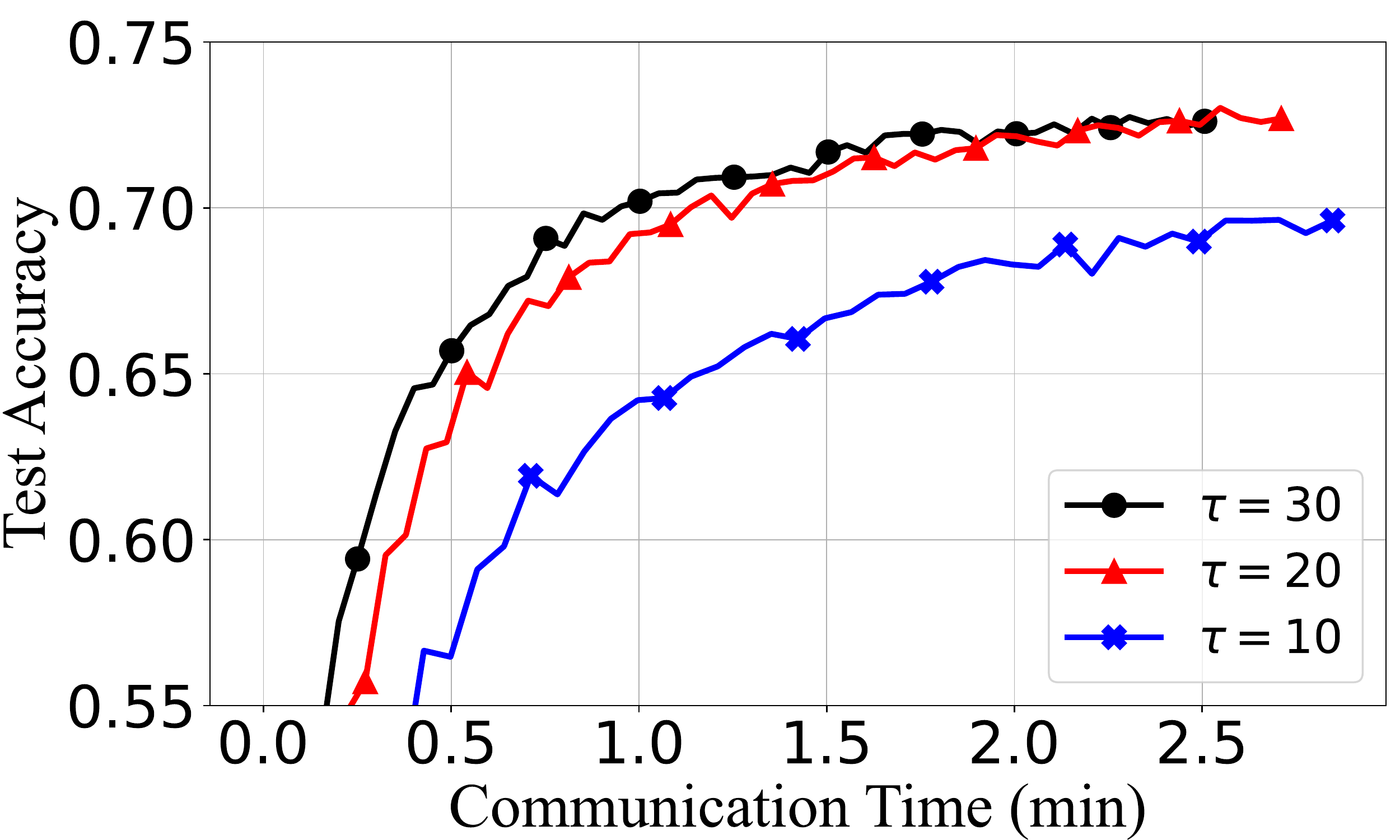}
}
\caption{
(a)~Test loss versus communication round with different numbers of prediction modes $N$. 
By increasing $N$, the residue error is reduced and the model will converge faster. 
(b)~In our proposed scheme, we use four modes and plot the frequency of different modes used during each communication round. 
(c)~Test accuracy versus communication time with different number of local iterations.
Given the communication time, $\uptau = 30$ gives the best test accuracy.  
}
\end{figure*}

\begin{table}[!tb]
    \centering
    \caption{\sc Simulation Parameters for Wireless Communication Channels }\label{table:sim_params}
    \begin{tabular}{cccc}
        \toprule
        Parameters & Value & Parameters & Value \\ \midrule
        $P_m$ & $0.01$ W  & $B$ & $2$ MHz \\ 
        $N_0$ & $-174$ dBm/Hz & $A_d$ & $4.11$ \\
        $f_c$ & $915$ MHz & $d_e$ & $2.8$ \\ \bottomrule
    \end{tabular}
\end{table}

For the simulations, we assume a circular area with a radius of $500$ meters, with one central server at its center and 30 uniformly distributed workers. 
The simulation parameters are listed in \papertab{table:sim_params}. 
We consider concurrent communication and evaluate the longest transmission time as the uplink communication cost.

\highlight{Model and Datasets.} We consider the learning tasks of image recognition with two datasets, namely, Fashion-MNIST~\cite{xiao2017fashion} and CIFAR-10~\cite{krizhevsky2009learning}.
Both of them  contain $C=10$ categories for classification. 
We followed Hsu et al. \cite{hsu2019measuring} and used the Dirichlet distribution to simulate the non-i.i.d. worker data.  
In particular, for the $m$th worker, we drew a random vector $\boldsymbol{q}_m \sim \text{Dir}(\alpha)$,  
where $\boldsymbol{q}_m = [q_{m,1}, \cdots, q_{m,C}]^{\top}$ belongs to the $(C-1)$-standard simplex.  
We then assigned the images to the worker, where the number of the $k$th class is proportional to $q_{m,k}$.   
We set $\alpha=0.5$ throughout the experiments. 
Since the non-i.i.d. setting is a defining characteristic for federated learning~\cite{mcmahan2017communication}, we focused on this non-i.i.d. data partition in the simulation.
More implementation details can be found in \paperApp{appendix:implementation}.

\highlight{Convergence.} 
We compare the proposed scheme with FedAvg\cite{mcmahan2017communication}, QSGD~\cite{alistarh2017qsgd}, STC~\cite{sattler2019robust}, FedPAQ~\cite{reisizadeh2020fedpaq}, and UVeQFed~\cite{shlezinger2020uveqfed}.
In the simulation, QSGD and FedPAQ adopt the stochastic scalar quantization and each entry is represented with $2$ bits. 
STC combines gradient sparsification and ternary quantization, where the sparsity rate is set to $1/400$. 
UVeqFed leverages the two-dimensional hexagonal lattice for gradient vector quantization, in which we set the quantization rate to $2$. 
We first remove the entropy coding modules in all methods and evaluate the test loss reduction over cumulative communication cost in \paperfig{fig:comparison}. 
It can be observed from \paperfig{fig:fmnist_loss_cost} and \paperfig{fig:cifar_loss_cost} that the proposed scheme outperforms all of them by achieving the lowest test loss given the same uplink cost. 
A more detailed comparison with the best two of these baseline algorithms reveals that the proposed scheme reduces the test loss by $30 \%$ compared to FedPAQ and UVeQFed on the Fashion-MNIST task given the communication cost of $0.8\times 10^{2}$--$1.2\times 10^{2}$ minutes,  
On the CIFAR-10 task, the proposed scheme reduces the test loss by $27 \%$ and $15 \%$ compared to FedPAQ and UVeQFed, respectively, given the communication cost of $4\times 10^{3}$--$6\times 10^{3}$ minutes.   
The results confirm our analysis in Remark 2, indicating that the interactive effect from the prediction and quantization improves the algorithm performance. 
In \paperfig{fig:comm_time}, we show the communication cost of different methods on the CIFAR-10 task for a test accuracy of $60 \%$. 
Compared to prior works such as FedAvg and UVeqFed, the proposed scheme can reduce the communication cost by two orders of magnitude.

\begin{table}[tb]
    \centering
    \caption{\sc Accuracy and Compression Ratio with Different Local Iterations}
    \label{table:effect_tau}
    \begin{tabular}{p{0.06\textwidth}p{0.09\textwidth}p{0.09\textwidth}p{0.09\textwidth}}
        \toprule
        Local Steps & Test Accuracy ($\%$) &  Compression Ratio & Communication Time (min) \\ \midrule
        $\uptau=10$ & $70.2 \pm 0.1 $ & $901\times$ & $3.5$  \\[3pt]
        $\uptau=20$ & $72.7 \pm 0.1 $ & $1183\times$ &  $2.7$ \\[3pt]
        $\uptau=30$ & $72.6 \pm 0.2 $ & $1279\times$  &  $2.5$ \\ \bottomrule
    \end{tabular}    
\end{table}

\highlight{Prediction Modes.}  In this experiment, we study the influence of the number of prediction modes. 
We train the model on CIFAR-10 for $50$ rounds and choose the number of predictor candidates $N$ from $1$ to $4$. 
FedPAQ can be viewed as a special case of our proposed scheme with one fixed prediction mode. 
The learning curves are plotted in \paperfig{subfig:loss_modes}. 
It can be observed that using more prediction modes can accelerate the learning. 
The observations confirm our analysis in Lemma~\ref{lemma:multimode_prediction} that increasing the number of prediction modes reduces the prediction error and hence accelerates the convergence rate.
For our proposed four-mode scheme, we visualize the different prediction modes selection frequency in \paperfig{subfig:freq_modes}. 
We note that mode 1 is a good candidate for the predictor, as its selection frequency is consistently around $0.5$.

\highlight{Local Iterations.} 
We study the effect of the number of local iterations $\uptau$. 
As $\uptau$ increases, the model update will have a larger variance and thus be difficult to predict. 
On the other hand, increased local computations will accelerate the model convergence to a certain extent, as we have shown in the first term of convergence bound~\eqref{eq:convergence_bound}.
A larger variance also means that the norm of the model updates is going to increase and the gradient will be sparser after quantization. 
We train the model on CIFAR-10 and select $\uptau$ from $\{10, 20, 30\}$.  
The learning curves are plotted in \paperfig{fig:effect_tau} and the training results after $50$ communication rounds are shown in \papertab{table:effect_tau}.  
It can be observed that when fixing communication time, a larger $\uptau$ can result in a higher test accuracy value. 
On the other hand, $\uptau=20$ gives the best test accuracy. 
The observations confirm our analysis that a larger $\tau$ can improve the communication efficiency by exploiting local computation at the workers, but may not necessarily improve the model accuracy due to the increased prediction errors, as we have pointed out in Remark~\ref{remark:eta_tau}.

\highlight{Quantization Levels.} 
In this experiment, we study how the number of quantization levels $2s + 1$ affects the learning procedure.  
As $s$ increases, the quantization error will be reduced and the model can achieve smaller errors based on \papertheorem{theorem:convergence}.
We train the model on CIFAR-10 and show the training results after $50$ communication rounds in \papertab{table:effect_s}. 
It can be observed that a larger $s$ improves the model accuracy at the cost of higher communication cost because of the increased quantization precision. 

\begin{table}[tb]
    \centering
    \caption{\sc Accuracy and Compression Ratio with Different Quantization Levels}
    \label{table:effect_s}
    \begin{tabular}{p{0.08\textwidth}p{0.09\textwidth}p{0.09\textwidth}p{0.09\textwidth}}
        \toprule
        Quantization Parameter & Test Accuracy ($\%$) & Compression Ratio & Communication Time (min) \\ \midrule
        $s=1$ & $72.7 \pm 0.2 $ &   $1183\times$ &  $2.7$ \\[3pt]
        $s=2$ & $74.1 \pm 0.2 $ &   $653\times$  &  $4.9$ \\[3pt]
        $s=3$ & $74.8 \pm 0.1 $ &   $463\times$  &  $6.9$  \\ \bottomrule
    \end{tabular}    
\end{table}

\section{Conclusion}\label{conclusion_section}
In this paper, we have focused on the communication-efficient federated learning and have proposed a predictive coding-based compression scheme. 
To the best of our knowledge, we are among the first to solve the task by jointly leveraging different compression tools, including predictive coding, quantization, and entropy coding.  
We have designed different prediction functions and let the worker choose the predictor dynamically to improve the compression performance. 
Our proposed scheme can significantly reduce the required bandwidth and communication cost and achieve better performance compared with other baselines, which has been confirmed by our empirical study. 

\makeatletter
\renewcommand{\subsection}{
\@startsection{subsection}{2}{0mm}
{-\baselineskip}
{0.5\baselineskip}{}}
\makeatother

\setcounter{Lemma}{0}
\setcounter{Theorem}{0}


\appendices

\section{Proof of Lemma \ref{lemma:multimode_prediction}}\label{proof:multimode_prediction}

\begin{proof}
        With the mode selection scheme, the final prediction error is calculated as $Y_{N} = \min (X_1, \dots, X_{N})$. 
        The cumulative density function (cdf) of $Y_{N}$ is 
        \begin{equation}
            F_{Y_N}(c) = 1 - \prod_{i=1}^N \prob[X_i \geqslant c] = 1 - \prod_{i=1}^N  [1 - F_{X_i}(c) ]. 
        \end{equation}
        Since $X_i$ is nonnegative, the complement cdf $1 - F_{X_i}(c)$ has $1$ for $c<0$ and is monotonic decreasing for $c>0$. 
        If one increases $N$ to $N+1$, $\prod_{i=1}^{N+1}  [1 - F_{X_i}(c) ]$ will be more discounted.
        Hence, $F_{Y_N}(c)$'s curve is rising, which tells us that the expectation
        \begin{equation}
            \expect[Y_N] =\int_{0}^{\infty}\left[ 1-F_{Y_N}(c) \right] \diff c
        \end{equation}
        is decreasing as $N$ increases.
\end{proof}

We empirically plot the prediction error versus the number of prediction modes $N$ in \paperfig{fig:error_N}. 
For simplicity, we assume that $X_i \sim \mathcal{N}(\mu, \sigma^2)$ and set $\mu=400$ and $\sigma=30$ in the simulation so that it is nearly impossible for $X_i$ to be negative. 
Each point on the curve is obtained by averaging over $1\times 10^{4}$ repetitions.
It can be observed that prediction error decreases as the number of predictor candidates increases.
Such decrease is the fastest for a small $N$.

\begin{figure}[tb]
    \centering
    \includegraphics[width=0.48\textwidth]{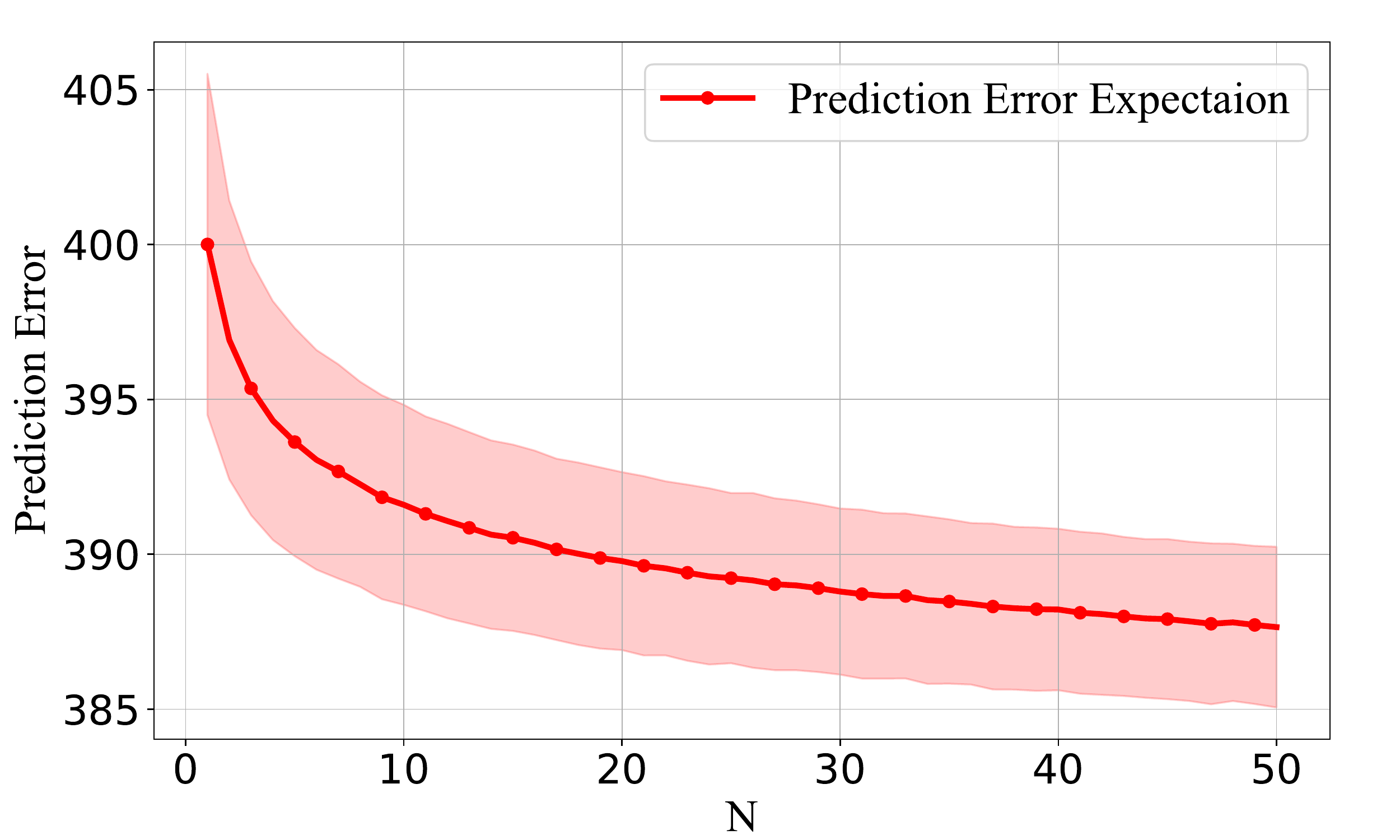}
    \caption{
    Prediction error expectation versus the number of predictor candidate modes $N$.
    The predictor error is simulated as $X_i \sim \mathcal{N}(\mu, \sigma^2)$ with $\mu=400$ and $\sigma=30$.
    The solid curve represents the mean value and the shaded region spans one standard deviation obtained over $1\times 10^{4}$ repetitions.
    The error expectation is a monotonically decreasing function of $N$, which is consistent with our analysis.   
    }
    \label{fig:error_N}
\end{figure}

\section{Proof of Lemma \ref{lemma:sign_mapping}}\label{proof:lemma:sign_mapping}

\begin{Lemma}
    Define $\phi_i \triangleq \sign(e_i) \cdot \varphi_i(\res, s). $
    The mapping defined in \eqref{eq:quant_map} can be rewritten as 
    \begin{equation}
        h_i = \left\{
        \begin{array}{l @{\;,\;} l}
            -2\, \phi_i   & \textrm{ if } \sign(\phi_i) \leqslant 0, \\
            2 \,  \phi_i - 1 & \textrm{ if } \sign(\phi_i) > 0.
        \end{array}
        \right.
    \end{equation}
    Consider following nonnegative integer encoding schemes: 
    (i) use the mapping defined in \eqref{eq:re_quant_map} and then encode $\phi_i$; 
    (ii) separately encode $|\phi_i|$ and $\sign (\phi_i)$.
    Scheme (i) has a shorter average codeword length.   
\end{Lemma}

\begin{proof}
    Consider $|\phi_i| \in \{0, 1, \cdots, J\}$. 
    We define the probability mass function of $|\phi_{i}|$ as follows:
    \begin{subnumcases}{\mathbb{P}(|\phi_i| = j) =}
         2 \, p_{j}, & \text{if } $j \neq 0$, \\
         p_0, & \text{otherwise }. 
    \end{subnumcases}
    The probabilities sum up to $1$, which yields
    \begin{equation}\label{eq:pmf_sumup_one}
        \sum_{j=1}^{J} 2 p_{j} + p_{0} = 1. 
    \end{equation}
    Under Assumption \ref{assumption:prediction}, it can be shown that
    \begin{equation}
        \mathbb{P} (\phi_i = j) = \mathbb{P} (\phi_i = -j) = p_{j}. 
    \end{equation}
    We use the entropy lower bound to estimate average codeword length.
    For Scheme (i) which maps $\phi_i$ to $h_i$ with \eqref{eq:re_quant_map}, the entropy of $h_i$ is calculated as  
    \begin{equation}
        H(h_i) =  -\left(\sum_{j=1}^{J} 2 p_{j} \log  p_{j} + p_{0} \log p_{0}\right).
    \end{equation}
    The average codeword length is estimated as $L_1 =  H(h_i)$.
    For Scheme (ii), the entropy of $|\phi_i|$ is calculated as 
    \begin{equation}
        H(|\phi_i|) =  -\left(\sum_{j=1}^{J} 2 p_{j} \log 2 p_{j} + p_{0} \log p_{0}\right).
    \end{equation}
    The average codeword length is estimated as $L_2 = H(|\phi_i|) + 1$. Here, the extra $1$ bit comes from the sign of $|\phi_i|$. 
    Taking the difference between $L_2$ and $L_1$ yields 
    \begin{subequations}
    \begin{align}
        L_2 - L_1 & = -\sum_{j=1}^{J} 2 p_{j} \log _{2} \frac{2 p_{j}}{p_{j}} + 1 \\
        &  \overset{\cirone}{=} p_0 > 0. 
    \end{align}
    \end{subequations}
    where $\cirone$ follows from \eqref{eq:pmf_sumup_one}. 
    The proof is complete. 
\end{proof}
\section{Proof of \papertheorem{theorem:convergence}}\label{appendix:proof_convergence}

\begin{proof}
For clarity, we first have a brief review of some key points in Algorithm \ref{coding_algorithm}. 
We denote the global weight difference as 
\begin{equation}
    \hde[k] \triangleq \w[k] - \w[k+1]. 
\end{equation}
According to \eqref{def:compressed_agg}, we have 
\begin{equation}\label{eq:hatdelta_expansion}
    \hde[k] =  \w[k] - \frac{1}{M} \sum_{m=1}^M \hw[k,\uptau]_{m}. 
\end{equation}
Substituting the decoding procedure $\hw[k+1]_{m} \leftarrow \pw[k+1] +  \hres[k]_{m}$ into \eqref{eq:hatdelta_expansion} yields 
\begin{equation}
    \hde[k] = \w[k]  - \frac{1}{M} \sum_{m=1}^M \left( \pw[k,\uptau]_{m} + \hres[k]_{m} \right) .
\end{equation}
We then take the expectation of $\hde[k]$.
Under Assumption \ref{assumption:prediction}, it can be shown that the quantization is unbiased, i.e.,
\begin{equation}\label{eq:quant_unbiased}
\mathbb{E}\left[ \hres[k]_{m} \Big| \res[k]_{m} \right] = \res[k]_{m},
\end{equation}
we can obtain
\begin{equation}\label{eq:def_global_delta}
    \de[k] \triangleq \mathbb{E}\left[ \hde[k] \Big| \w[k] \right] = \w[k] - \frac{1}{M} \sum_{m=1}^{M} \w[k,\uptau]_{m}. 
\end{equation}

From Assumption \ref{assumption:Lsmooth}, we have 
\begin{multline}
    \mathbb{E} \left[f(\w[k+1]) - f(\w[k])\right]  \leqslant \\
    - \underbrace{\mathbb{E} \innprod{\nabla f(\w[k])}{\hde[k]}}_{T_1} 
    + \frac{L}{2}  \underbrace{\expect \normsq{\hde[k]}}_{T_2}. \label{eq:original_inequality}
\end{multline}  
We first consider the bound for $T_1$:
\begin{subequations}
\begin{align}
  &\; \expect \innprod{\nabla f(\w[k])}{\hde[k]} \\
\overset{\cirone}{=} &\; \expect \innerprod{\nabla f(\w[k])}{\de[k]} \\
= &\; \expect \innerprod{\nabla f(\w[k])}{\frac{\eta}{M} \sum_{m=1}^{M}\sum_{t=0}^{\uptau-1} \nabla f_{m}(\w[k, t]) } \\
\overset{\cirtwo}{=} &\; \frac{\eta}{2 M} \sum_{m=1}^{M} \sum_{t=0}^{\uptau-1} 
       \Big(\expect \normsq{\nabla f(\w[k])} + \expect \normsq{\nabla f(\mathbf{w}_{m}^{(k,t)})} \nonumber\\ 
    &\;    - \expect \normsq{\nabla f(\w[k]) - \nabla f(\w[k,t]_{m})} \Big) \\
\overset{\cirthree}{\geqslant}  &\; \frac{\eta \uptau}{2} \, \expect \normsq{\nabla f(\w[k])} 
      + \frac{\eta}{2M} \sum_{m=1}^{M} \sum_{t=0}^{\uptau-1} \Big(\expect \normsq{\grad[k,t]} \nonumber\\
    &\;  - L^2 \underbrace{\expect \normsq{\w[k,t]_{m} - \w[k]}}_{T_3} \Big), \label{eq:inner_product_bound}
\end{align}
\end{subequations}
where $\cirone$ follows from \eqref{eq:quant_unbiased}, 
$\cirtwo$ follows from $2 \langle \vecA, \vecB \rangle = \|\vecA\|_2^2 + \|\vecB\|^2_2 - \|\vecA - \vecB\|^2_2$, 
and $\cirthree$ follows from Assumption \ref{assumption:Lsmooth}. 
$T_3$ is bounded as 
\begin{subequations}
\begin{align}
&\; \expect  \normsq{\w[k,t]_{m} - \w[k]} \nonumber \\ 
= &\; \eta^2 \expect \normsq{\sum_{i=0}^{t-1} \left(\stograd[k,i] - \grad[k,i]\right)} 
    + \expect \normsq{\sum_{i=0}^{t-1} \grad[k,i]} \nonumber \\
    &\; + 2  \expect \innerprod{\sum_{i=0}^{t-1} \left(\stograd[k,i] - \grad[k,i] \right)}{\sum_{i=0}^{\uptau-1} \grad[k,i]}  \\
\overset{\cirone}{=} &\;  \eta^2 \expect \normsq{\sum_{i=0}^{t-1} \left(\stograd[k,i] - \grad[k,i]\right)} 
    + \expect \normsq{\sum_{i=0}^{t-1} \grad[k,i]} \\
\overset{\cirtwo}{\leqslant} &\; \eta^2 t\,\left(\sum_{i=0}^{t-1} \expect \normsq{\grad[k,i]} + \sigma_{\xi}^2\right), \label{eq:T3}  
\end{align}
\end{subequations}
where $\cirone$ and $\cirtwo$ follow from Assumption \ref{assumption:sto_grad}.
Plugging \eqref{eq:T3} into \eqref{eq:inner_product_bound} yields 
\begin{subequations}

\begin{align}
    &\; \expect \innprod{\nabla f(\w[k])}{\hde[k]} \geqslant \frac{\eta \uptau}{2} \, \expect \normsq{\nabla f(\w[k])}  \nonumber \\
    &\; + \frac{\eta}{4 M}\left(2-L^{2} \eta^{2} \uptau(\uptau-1)\right) \sum_{m=1}^{M} \sum_{t=0}^{\uptau-1} \expect \normsq{\grad[k,t]} \nonumber \\
    &\; - \frac{L^{2} \eta^{3} \uptau(\uptau-1)}{4} \sigma_{\xi}^{2}. \label{eq:T1}
    \end{align}
\end{subequations}
Next, we consider the bound for $T_2$:
\begin{subequations}
\begin{align}
&\; \expect \normsq{\hde[k]} \\
= &\; \expect \normSQ{\w[k]  - \frac{1}{M} \sum_{m=1}^M \left(\pw[k,\uptau]_{m} + \hres[k]_{m}\right)}  \\
= &\; \expect \normSQ{\w[k]  \! -\! \frac{1}{M} \sum_{m=1}^M \! \left(\pw[k,\uptau]_{m} + \res[k]_{m} + \hres[k]_{m} - \res[k]_{m}  \right) } \\
\overset{\cirone}{=} &\; \expect \normSQ{\w[k] \! - \! \frac{1}{M} \! \sum_{m=1}^M \! \w[k,\uptau]_{m} } 
    \! + \expect \normSQ{ \frac{1}{M} \sum_{m=1}^M \left( \hres[k]_{m} - \res[k]_{m}\right)} \\
\overset{\cirtwo}{\leqslant} &\; \expect \normSQ{\frac{\eta}{M} \sum_{m=1}^M \sum_{t=0}^{\uptau-1} \stograd[k,t] } 
    + \frac{1}{M^2} \normSQ{\sum_{m=1}^{M} \left( \hres[k]_{m} - \res[k]_{m} \right) } \\
\overset{\cirthree}{\leqslant} &\; \expect \normSQ{\frac{\eta}{M} \sum_{m=1}^M \sum_{t=0}^{\uptau-1} \stograd[k,t] } 
    + \frac{q}{M^2} \sum_{m=1}^{M}\normsq{\res[k]_{m}} \\
\overset{\cirfour}{\leqslant} &\; \frac{\eta^2\uptau}{M} \sum_{m=1}^M \sum_{t=0}^{\uptau-1} \expect  \normsq{\grad[k,\uptau]} 
    + \frac{\eta^2\uptau}{M} \sigma_{\xi}^2
    + \frac{q}{M^2} \underbrace{\sum_{m=1}^{M}\normsq{\res[k]_{m}}}_{T_4}.  \label{eq:T2_intermediate}
%
\end{align}
\end{subequations}
where $\cirone$ follows from \eqref{eq:quant_unbiased}, $\cirtwo$ follows from Assumption \ref{assumption:quantization}, $\cirthree$ follows from Assumption \ref{assumption:quantization} and i.i.d. quantization noise, 
and $\cirfour$  follows from Assumption \ref{assumption:sto_grad}. 
Under Assumption \ref{assumption:prediction}, we bound $T_4$ as 
\begin{subequations}
\begin{align}
 \sum_{m=1}^{M}\normsq{\res[k]_{m}} 
& \leqslant  \sum_{m=1}^{M} p \normsq{\de[k]_{m}} \\
& \leqslant  \sum_{m=1}^{M} p \normsq{\eta \sum_{t=0}^{\uptau-1} \stograd[k,t] } \\ 
\leqslant & \; p\,\eta^2\uptau \sum_{m=1}^{M} \sum_{t=0}^{\uptau-1} \normsq{\stograd[k,t]} \\
& \leqslant   p\,\eta^2\uptau \left(\sum_{m=1}^{M} \sum_{t=0}^{\uptau-1} \normsq{\grad[k,t]} + M \sigma_{\xi}^2  \right). \label{eq:T4}
\end{align}
\end{subequations}
Plugging \eqref{eq:T4} into \eqref{eq:T2_intermediate} yields 
\begin{multline}
    \expect \normsq{\hde[k]} \leqslant 
    \frac{\eta^2\uptau}{M} \left( \frac{qp}{M} + 1\right) \sum_{m=1}^{M} \sum_{t=0}^{\uptau-1} \normsq{\grad[k,t]} \\ 
    + \frac{\eta^2\uptau}{M}(qp+1) \sigma_{\xi}^2.  \label{eq:T2}
\end{multline}
Substituting the results in \eqref{eq:T1} and \eqref{eq:T2} into \eqref{eq:original_inequality} yields 
\begin{subequations}
\begin{align}
&\; \mathbb{E} \left[f(\w[k+1]) - f(\w[k])\right] \nonumber \\
\leqslant &\; -\frac{\eta\uptau}{2} \expect \normsq{\nabla f(\w[k])} 
              - \frac{\eta}{4 M} \Big(2-L^{2} \eta^{2} \uptau(\uptau-1) \nonumber \\ &\; -2L\eta\uptau (\frac{qp}{M} + 1) \Big) \sum_{m=1}^{M}  \sum_{t=0}^{\uptau-1} \normsq{\grad[k,t]} \nonumber \\
          &\; + \frac{L \eta^{2} \uptau}{4}\left(\frac{2(qp+1)}{M}+L \eta(\uptau-1) \right) \sigma_{\xi}^{2} \\
\overset{\cirone}{\leqslant} &\; -\frac{\eta\uptau}{2} \expect \normsq{\nabla f(\w[k])}  \nonumber \\
&\; + \frac{L \eta^{2} \uptau}{4}\left(\frac{2(qp+1)}{M} + L \eta(\uptau-1)\right) \sigma_{\xi}^{2},
\end{align}
\end{subequations}
where $\cirone$ follows from the constraint for $\eta$ in \eqref{eq:constraint_lr}. 
Summing up over $K$ communication rounds yields 
\begin{multline}
    \frac{1}{K} \sum_{k=0}^{K-1} \expect \normsq{\nabla f(\w[k])} \leqslant 
    \frac{2\left(f\left(\w[0]\right)-f\left(\w^*\right)\right)}{\eta \uptau K} \\
    + L \eta\left(\frac{qp+1}{M} + \frac{L \eta(\uptau-1)}{2}\right) \sigma_{\xi}^{2}.  
\end{multline}
\end{proof}

\section{Implementation}\label{appendix:implementation}
We used a LeNet model for Fashion-MNIST task and a VGG-7 model~\cite{simonyan2015very} for the CIFAR-10 task. 
For the worker local update, we used the Adam optimizer~\cite{kingma2015adam} and searched the learning rate over the set $\{10^{-4}, 5\times 10^{-4}, 10^{-3}, 5\times 10^{-3}, 10^{-2}\}$.  
For prediction mode 2, we set the step size to $1\times10^{-3}$. 
For prediction mode 3, we set the order $R$ to $3$. 
For prediction mode 4, we use $\beta_1=0.8$ and $\beta_2=0.99$. 
The implementation is available at \url{https://github.com/KAI-YUE/Predictive-Coding-FL}.

\bibliographystyle{IEEEtran}
\bibliography{ref}

\end{document}